%% file: main.tex
\documentclass[prl,twocolumn,superscriptaddress,floatfix]{revtex4-2}

\usepackage{amsmath,amssymb,amsthm}
\usepackage{braket}
\usepackage{graphicx}
\usepackage{booktabs}
\usepackage{hyperref}
\usepackage{color}
\usepackage{enumitem}

\newcommand{\supplM}{Appendix}

\graphicspath{{figures/}}

\newcommand{\inp}[2]{\langle #1 | #2 \rangle}
\newcommand{\outp}[2]{| #1 \rangle \langle #2 |}

\newcommand{\proj}[1]{| #1 \rangle \langle #1 |}
\newcommand{\Complex}{\mathbb{C}}
\newcommand{\E}{\mathbb{E}}
\newcommand{\tr}{\operatorname{tr}}

\newtheorem{theorem}{Theorem}
\newtheorem{lemma}[theorem]{Lemma}
\newtheorem{proposition}[theorem]{Proposition}
\newtheorem{corollary}[theorem]{Corollary}

\theoremstyle{definition}
\newtheorem{definition}[theorem]{Definition}
\newtheorem{example}[theorem]{Example}
\theoremstyle{remark}
\newtheorem{remark}[theorem]{Remark}

\begin{document}

\title{Quantum Advantage for Coordinated Frequency Selection Against Distributed Jammers}

\author{Stephanie Wehner}
\email{s.d.c.wehner@tudelft.nl}
\affiliation{QuTech, Delft University of Technology, 2628~CJ Delft, The Netherlands}
\affiliation{Quantum Computer Science, EEMCS, Delft University of Technology,
Lorentzweg 1, 2628 CJ Delft, The Netherlands}
\affiliation{Kavli Institute of Nanoscience, Delft University of Technology, Lorentzweg 1, 2628 CJ Delft, The Netherlands}

\date{\today}

\begin{abstract}
Consider two parties who want to agree on a common frequency band for communication in the presence of independent
jammers. Such jammers block a different subset of bands at each site, where each party can observe only its own set of unjammed  bands.
Yet, they must agree on a common band without communicating. We first establish the optimal classical strategy, maximizing
the probability they output a common frequency band in a single shot. We proceed to show that sharing an entangled pair of local dimension~$d$
allows the parties to coordinate strictly better, provided both the
number of safe bands~$d$ and the spectrum size~$n$ are sufficiently large. 
We study explicit quantum strategies offering a pathway to near-term demonstrations, 
including an explicit strategy for $d = 2$
that outperforms the classical optimum for all spectrum sizes, achieving a $5.4\%$ advantage asymptotically (in $n$) with
just one shared Bell pair. Our approach is based on a general framework for constructing quantum strategies from
classical spreading sequences via symmetric orthonormalization that may be of independent interest, and opens the door to concrete applications of quantum networks for cognitive radio and spread-spectrum communication.
\end{abstract}

\maketitle

%%% INTRODUCTION (unnumbered, flowing) %%%

A quantum network can be used to distribute entanglement between distant
quantum memories, creating a resource that persists until
needed~\cite{Wehner2018, Kimble2008}.  Entanglement that is stored at distant quantum nodes may subsequently be 
measured to obtain non-local correlations that are strictly stronger than what is possible classically, an effect known as Bell 
non-locality~\cite{Brunner2014,Bell1964}.
Recent work has begun to use non-local correlations for practical applications for low-latency coordination. Here, measurements are chosen depending on local observations such as customer demands or the behaviour of financial markets. In turn, the outcomes of such measurements are
used to make decisions, such as which servers to route customer demands to, or what market actions to take. When entanglement has been pre-shared, this can allow a (near) instantaneous form of coordination that exploits the fundamental concept of Bell-nonlocality to beat any classical method.  Example applications include load balancing~\cite{DaSilva2025,Arun2025}, distributed control~\cite{Nguyen2013}, rendezvous problems on graphs~\cite{Mironowicz2023,Viola2024,Tucker2024}, high-frequency trading~\cite{Telepathy2024}, decentralized team coordination~\cite{DeshpandeKulkarni2023}, and entanglement-assisted communication over multiple-access channels~\cite{Leditzky2020,Notzel2020}. We refer to~\cite{DingXu2026} for a recent overview of coordination
applications of nonlocality, and to~\cite{LatencyConstraints2025}
for an extension where some communication is allowed in addition to
nonlocal correlations.

In this work, we introduce and analyze a new coordination problem and show that
entanglement provably helps: \emph{spectrum allocation under
adversarial jamming}.  Consider $n$ frequency bands shared between two
spatially separated parties, Alice and Bob, who wish to communicate on
a common band.  Each site faces an independent jammer that blocks $k$
of the $n$ bands.  Alice and Bob each observe their local set of $d =
n - k$ safe (unjammed) bands, but cannot see which bands are safe at
the other site.  They wish to select a common frequency band without communicating during the process. 
That is, Alice and Bob succeed if and only if they choose the \emph{same} band and
it is safe at \emph{both} sites.  We call this the $(n,k)$-jamming
game (Fig.~\ref{fig:game}).

We focus on the single-shot success probability
\begin{equation}\label{eq:omega}
    \omega = \frac{1}{|\mathcal{S}|^2}
    \sum_{x,y \in \mathcal{S}} \sum_{c \in x \cap y} P(c,c \mid x,y)\,,
\end{equation}
where $\mathcal{S} = \{S \subseteq [n] : |S| = d\}$ with $[n]=\{0,\ldots,n-1\}$ collects all
possible safe sets, and $P(a,b \mid x,y)$ is the probability that
Alice outputs~$a$ and Bob outputs~$b$ given safe sets~$x$ and~$y$. 
We here take $x$ and $y$ to be chosen uniformly at random, modelling unknown adversarial jamming.
Our single-shot performance metric $\omega$ is applicable to
fast-fading environments where the jamming pattern changes
independently between rounds, so maximizing $\omega$ is equivalent to
minimizing the expected time to
rendezvous~\cite{Bian2006}, i.e. the number of rounds until the parties
first communicate successfully.  The jamming game is closely related to the blind rendezvous problem
in cognitive radio~\cite{Shin2010, Lin2011, BlindRendezvousSurvey2022},
with the adversarial jammers playing the role of primary users. We remark that quantum game
theory has been applied to cooperative spectrum
sharing~\cite{Zabaleta2017}, but in a fundamentally different setting where players apply unitaries to a shared state controlled by a referee,
whereas here the game itself is classical and the advantage arises from
pre-shared entanglement measured locally.

Let us now consider different strategies that the two parties might employ in order to maximize the success probability~\eqref{eq:omega}.
A \emph{classical strategy ($cl$)} consists of shared randomness, also known as local-hidden variables~\cite{Brunner2014}.
When maximizing the probability $\omega_{cl}$ to win using a classical strategy, we can assume a deterministic strategy (Eq.(25) in~\cite{Brunner2014}), which can be expressed as functions
$f,g\colon \mathcal{S} \to [n]$ for that Alice ($f(x)=a$) and Bob ($g(y)=b$) use to map a safe set to a choice of frequency band. A \emph{quantum strategy} additionally uses a shared entangled state
$\ket{\Psi}$ and local measurements: upon receiving safe set~$x$,
Alice measures her half of $\ket{\Psi}$ with a measurement depending on $x$ to choose a channel, and Bob
does the same with his half.  The resulting correlations can be written as $P(a,b|x,y) =
\bra{\Psi} M_a^x \otimes N_b^y \ket{\Psi}$, in terms of measurement operators $\{M_a^x\}_a$ for Alice, and $\{N_b^y\}_b$ for Bob.

We now first establish the optimal classical winning probability (Theorem~1).
We then present a general framework for constructing quantum strategies from seed vectors
via L\"owdin orthonormalization which may be of independent interest. Using random seed vectors, we show 
entanglement provides an advantage for sufficiently large $d$ and $n$ (Theorem~2). Finally, we turn to explicit, deterministic seed frame
constructions, demonstrating near-term feasibility (Theorem 3). Full technical details are provided in~\supplM, software for the numerical calculations, as well as /using Lean formalized proofs can be found at~\cite{JammingNumerics}.

%%% FIGURE 1: Game illustration %%%
\begin{figure}[t]
    \centering
    \includegraphics[width=\columnwidth]{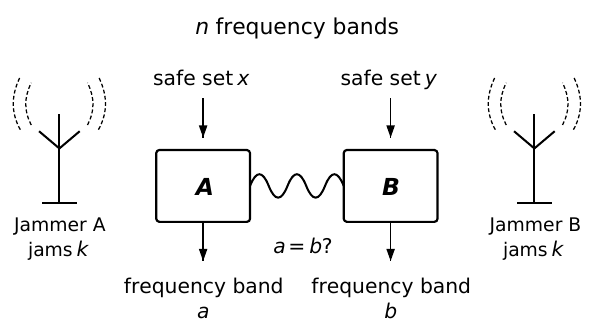}
    \caption{The $(n,k)$-jamming game. Alice and Bob each observe $d = n-k$
    safe frequency bands (set of unjammed bands).  They output a choice of frequency band $a$ and $b$ respectively, and succeed if they agree ($a=b$).  The jammers act
    independently and uniformly at random.}
    \label{fig:game}
\end{figure}

%%% FIGURE 3: Heatmap %%%
\begin{figure*}[tp]
    \centering
    \includegraphics[width=\textwidth]{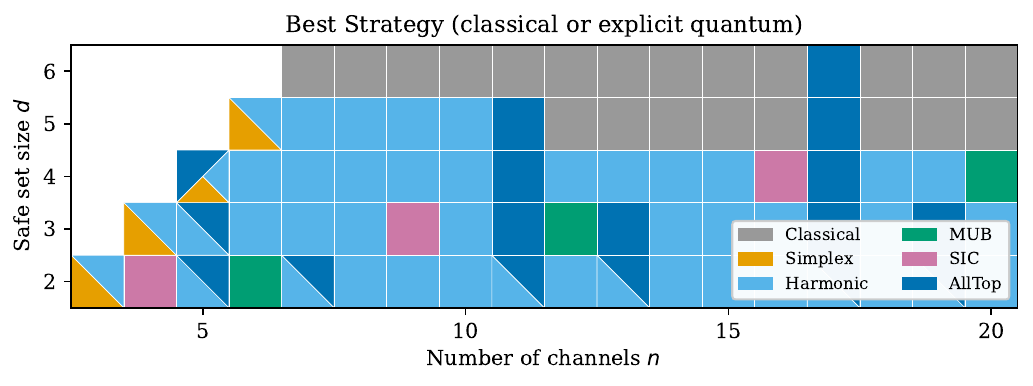}
    \caption{Best explicit quantum strategy at each $(n,d)$ point.  Cells
    are colored by the frame achieving the highest winning probability.
    The harmonic frame dominates the interior, while structured frames
    (MUB, SIC, simplex) are optimal when they exist. %Cells marked ``None'' indicate no quantum advantage from
   %explicit frames.
   }
    \label{fig:heatmap}
\end{figure*}

%%% TABLE: Frame comparison %%%
\begin{table*}[tp]
\centering
\small
\begin{tabular}{ccc|c|ccccc|c}
\toprule
$n$ & $d$ & $k$ & $\omega_{cl}$ & Simplex & Harmonic & MUB & SIC & AllTop & Opt.\ \\
\midrule
3 & 2 & 1 & 0.5556 & \textbf{0.5833} & 0.5833 & -- & -- & -- & 0.5833 \\
4 & 2 & 2 & 0.3889 & -- & 0.4119 & -- & \textbf{0.4167} & -- & 0.4167 \\
5 & 2 & 3 & 0.3000 & -- & 0.3184 & -- & -- & 0.3184 & \textbf{0.3230} \\
6 & 2 & 4 & 0.2444 & -- & 0.2595 & \textbf{0.2644} & -- & -- & 0.2644 \\
\midrule
4 & 3 & 1 & 0.6250 & \textbf{0.6451} & 0.6451 & -- & -- & -- & 0.6451 \\
5 & 3 & 2 & 0.4600 & -- & 0.4767 & -- & -- & 0.4767 & \textbf{0.4808} \\
6 & 3 & 3 & 0.3650 & -- & 0.3783 & -- & -- & -- & \textbf{0.3852} \\
\midrule
5 & 4 & 1 & 0.6800 & \textbf{0.6909} & 0.6909 & -- & -- & 0.6909 & 0.6909 \\
6 & 4 & 2 & 0.5200 & -- & 0.5290 & -- & -- & -- & \textbf{0.5333} \\
\midrule
6 & 5 & 1 & 0.7222 & \textbf{0.7260} & 0.7260 & -- & -- & -- & 0.7260 \\
7 & 6 & 1 & 0.7551 & 0.7537 & 0.7537 & -- & -- & -- & \textbf{0.7537} \\
\bottomrule
\end{tabular}
\caption{Quantum winning probabilities for different seed frame constructions.
Bold indicates the best value for each $(n,d)$.
A dash (--) means the frame does not exist for those parameters.
The last column (Opt.)\ reports the numerically optimized seed frame~\cite{JammingNumerics}.
At $(n,d) = (7,6)$, even the optimized quantum value falls below~$\omega_c$.}
\label{tab:frame-comparison}
\end{table*}
\smallskip
%%% SECTION I: CLASSICAL VALUE %%%
\textit{Classical winning probability:}
In order to claim a quantum advantage, we first establish how well a classical strategy can perform at the jamming game.

\textbf{Theorem 1}---\textit{The maximum success
probability attainable by any classical strategy is}
\begin{equation}\label{eq:classical}
    \omega_{cl}(n,k)
    = \frac{1}{\binom{n}{k}^2}
      \sum_{i=0}^{k} \binom{n\!-\!1\!-\!i}{k-i}^{\!2}\,.
\end{equation}

We establish our result in three steps (see \supplM\ for details):

Step 1 ($f=g$):
First of all, note that in~\eqref{eq:omega} for deterministic classical strategies $P(c,c|x,y)=1$ if and only if $f(x) = g(y) = c$.
Since outputting an unsafe band never wins, we may assume $f(x) \in x$ and $g(y) \in y$,
so that $\omega_{cl}$ reduces to counting pairs
$|\{(x,y) : f(x) = g(y)\}|$.
For a fixed output~$c$, the number of such pairs is
$|f^{-1}(c)|\,|g^{-1}(c)|$, where $f^{-1}(c) = \{x : f(x) = c\}$
and similarly for~$g$. Thus summing over~$c$ and normalizing,
\begin{align}
  \omega_{cl} = \frac{1}{|\mathcal{S}|^2}
  \sum_{c \in [n]} |f^{-1}(c)|\,|g^{-1}(c)|\,.
\end{align}
By Cauchy--Schwarz,
$\sum_c |f^{-1}(c)|\,|g^{-1}(c)|
\leq \sqrt{\sum_c |f^{-1}(c)|^2}\,\sqrt{\sum_c |g^{-1}(c)|^2}$,
with equality when the two vectors are proportional.
Since both satisfy
$\sum_c |f^{-1}(c)| = \sum_c |g^{-1}(c)| = |\mathcal{S}|$,
proportionality forces $|g^{-1}(c)| = |f^{-1}(c)|$ for all~$c$,
so maximizing $\omega_{cl}$ allows choosing $g = f$.

Step 2 (Saturation): Setting $g = f$, it remains to maximize 
$\sum_c |f^{-1}(c)|^2$ subject to $f(x) \in x$. Since each band~$c$ 
belongs to exactly $\binom{n-1}{k}$ safe sets, we have 
$|f^{-1}(c)| \leq \binom{n-1}{k}$. If the largest $|f^{-1}(c)|$ is 
below this bound, there exists a safe set~$x$ with $c \in x$ but 
$f(x) \neq c$. Changing $f(x)$ to~$f(x)=c$ moves~$x$ from a smaller to a 
larger $f^{-1}$, which by convexity of $t \mapsto t^2$ strictly increases $\sum_c |f^{-1}(c)|^2$. 
Hence at optimum, the largest $|f^{-1}(c)|$ must equal 
$\binom{n-1}{k}$.

Step 3 (Greedy strategy): W.l.o.g. let band $c=1$  reach this value.
The  remaining safe sets (those not containing $1$) form the input space 
of a $(n\!-\!1, k\!-\!1)$-jamming game on bands $\{2, \ldots, n\}$. 
The strategy $f^*(x) = \min(x)$ (outputting the smallest band in~$x$) 
achieves this: it saturates channel $1$ and recursively solves each 
sub-game. Induction on~$k$ shows that $f^*$ is optimal and yields 
Eq.~\eqref{eq:classical}.

It is instructive to consider Theorem~1 for fixed~$d$ and large~$n$,
where the classical value decays as
$\omega_{cl}^* \sim d^2\big/\!\left((2d\!-\!1)\,n\right)$:
coordination becomes harder as safe sets become sparse subsets of a
growing spectrum. This is the baseline that quantum strategies will
need to beat.

%%% SECTION II: QUANTUM STRATEGIES %%%
\smallskip
\textit{Framework:}
We now present a general framework of constructing quantum strategies for the jamming game. 
First, since the jamming game is a synchronous game~\cite{PaulsenSeverini2016}, we know that the optimal shared state is a maximally entangled state
\begin{equation}
\ket{\Psi_d} = \frac{1}{\sqrt{d}} \sum_{j=0}^{d-1} \ket{j}\ket{j}\ .
\end{equation}
Motivated by the desire to allow an implementation using minimal hardware requirements, we choose a minimal local dimension $d =
n-k$, where $d$ is the number of measurement outcomes. We now let Alice's measurements be given by the bases
$\mathcal{A}_x = \big\{\ket{v_x^c} \mid c \in x\ \}$ for safe set $x \in \mathcal{S}$ with $|x| = d$.
We then choose for Bob the same bases as for Alice up to conjugation $\mathcal{B}_y = \big\{\ket{w_y^c} \mid c \in y\ \}$, where $\ket{w_s^c} = \ket{v_s^c}^*$. Using the identity $\bra{\Psi_d} A \otimes B\ket{\Psi_d} = \tr(AB^T)/d$~\cite{NielsenChuang2010}, we can then write the quantum success probability~\eqref{eq:omega} as 
\begin{equation}\label{eq:omega-q}
\omega_q  = \frac{1}{d |S|^2} \sum_{c \in [n]} \left[\sum_{\substack{x,y\\c \in x \cap y}}  \left| \inp{v_x^c}{v_y^c} \right|^2\right]\ ,
\end{equation}
where $\ket{v_s^c} = 0$ for all $c \notin s$ is the zero vector, allowing us to rearrange the sums.
We see that for any fixed $c$, the sum in brackets becomes large when the
measurement vectors for the \emph{same} band~$c$ across
\emph{different} safe sets point in a similar direction.  The challenge
in choosing Alice and Bob's measurement basis is thus to make these cross-context vectors 
as aligned as possible, while satisfying the requirements for a measurement basis, i.e. that the vectors
are orthonormal within each safe set.
Our construction resolves this tension in two physically motivated
steps:

\emph{Seed vectors.}  We assign one \emph{seed vector}
$\ket{\phi_c} \in \Complex^d$ to each channel $c \in [n]$. These can be thought of
 as ``ideal'' directions that encode channel identity.  Because
$n > d$, these $n$ vectors cannot all be orthogonal, so we choose a
frame (an overcomplete set spanning~$\Complex^d$) whose elements are
as uniformly spread as possible.

\emph{L\"owdin orthonormalization.}  Within each safe set~$x$
(containing exactly $d$ channels), the $d$ seed vectors can thus be
non-orthogonal. If the vectors span~$\Complex^d$, we orthonormalize them using
the L\"owdin (symmetric) procedure~\cite{Lowdin1950}, originally
developed in quantum chemistry to construct orthonormal molecular
orbitals from overlapping atomic orbitals while preserving their
spatial character as much as possible:
\begin{equation}\label{eq:lowdin}
    \ket{v_x^c} = \sum_{c' \in x}
    \!\left[G_x^{-1/2}\right]_{\!c',c}\ket{\phi_{c'}},
    \quad c \in x\,,
\end{equation}
where $G_x$ is the Gram matrix with entries $[G_x]_{c,c'} =
\inp{\phi_c}{\phi_{c'}}$.  Among all orthonormalization procedures,
L\"owdin's is the one that keeps each output vector closest to its
input~\cite{Lowdin1950,Mayer2003}, i.e. it distorts the seed directions as
little as possible.  This maximal preservation of the seed geometry is
what drives cross-context alignment and hence large $\omega_q$ in Eq.~\eqref{eq:omega-q}.
When the seed vectors within a safe set are not linearly independent
(so that $G_x$ is singular), we replace $G_x^{-1/2}$ by the
Moore--Penrose pseudoinverse $(G_x^+)^{1/2}$, yielding the Pretty
Good Measurement (PGM)~\cite{Belavkin1975,HausladenWootters1994}.
Figure~\ref{fig:bloch} illustrates this construction on the Bloch
sphere for~$d = 2$.

%%% FIGURE 2: Bloch sphere %%%
\begin{figure}[t]
    \centering
    \includegraphics[width=\columnwidth]{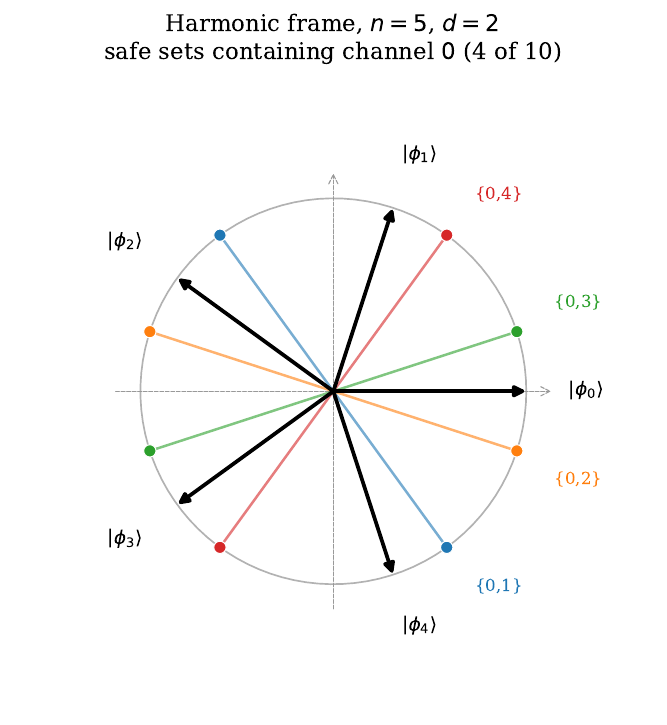}
    \caption{Bloch-sphere visualization of the quantum strategy for the
    harmonic frame with $n = 5$, $d = 2$.  Black arrows: seed vectors
    $\ket{\phi_c}$.  Colored lines: measurement bases (pairs of antipodal
    L\"owdin vectors) for six safe sets.  The orthonormalization rotates
    each pair toward the seed vectors while maintaining orthogonality,
    preserving cross-set alignment.}
    \label{fig:bloch}
\end{figure}

\emph{Quantum advantage for large~$d$:}
We now examine strategies that beat any classical strategy, by choosing appropriate seed vectors.
We first prove that entanglement
helps for all sufficiently large safe-set sizes, provided $n$ is large enough.  Drawing
seed vectors independently from the Haar measure on the unit sphere
in~$\Complex^d$, we establish:

\textbf{Theorem 2.}---\textit{There exists $d_0$ such that for all $d \geq d_0$, there
exists $N(d)$ such that for all $n \geq N(d)$,}
\begin{equation}\label{eq:advantage}
    \E[\omega_q] > \omega_{cl}^*(n,\, n\!-\!d)\,,
\end{equation}
\textit{where the expectation is taken over the Haar measure.
The asymptotic advantage ratio satisfies}
\begin{equation}\label{eq:ratio}
    \lim_{n \to \infty}
    \frac{\E[\omega_q]}{\omega_c^*}
    = \frac{(2d\!-\!1)\,\mathcal{L}_1}{d}\,,
\end{equation}
\textit{where $\mathcal{L}_1 = \alpha_d^2 + (1 - \alpha_d)^2/(d-1)$.}

Here, $\mathcal{L}_1$ is the expected squared L\"owdin overlap $\E[|\inp{v_X^c}{v_Y^c}|^2]$ between measurement vectors for the same channel~$c$ from two safe sets $X, Y$ sharing only that channel ($|X \cap Y| = 1$), averaged over the Haar measure on seed vectors. The alignment parameter is $\alpha_d := \E[\langle 0|S_X^{-1/2}|0\rangle^2]$, where $X$ is any safe set of size~$d$ containing channel~$0$, $S_X = \sum_{c \in X} \proj{\phi_c}$, and the expectation is over the $d-1$ Haar-random seed vectors in $X \setminus \{0\}$. $\alpha_d$ measures how much of the original seed direction survives L\"owdin orthonormalization on average. Note that since the expected value beats classical, we can conclude there exists at least one quantum strategy that beats classical. We establish our result in three steps (see \supplM\ for details):

Step 1 (Concentration): For fixed~$d$ and large~$n$,
the winning probability is dominated by pairs of safe sets sharing
a single element. This gives
$\E[\omega_q] \approx (d/n)\,\mathcal{L}_1$,
where $\mathcal{L}_1 = \E[|\inp{v_X^c}{v_Y^c}|^2]$ is the expected squared L\"owdin overlap for safe sets $X, Y$ sharing a single channel~$c$ (i.e., $|X \cap Y| = 1$), averaged over the Haar measure on seed vectors.

Step 2 (Computing $\mathcal{L}_1$): For safe sets sharing a
single element, the remaining seed vectors are drawn independently
from the Haar measure. Decomposing each L\"owdin vector into a
component along the seed direction (of magnitude~$\alpha_d$) and a
perpendicular remainder, and averaging gives
$\mathcal{L}_1 = \alpha_d^2 + (1 - \alpha_d)^2/(d-1)$.
The parameter~$\alpha_d$ equals the success probability of the
Pretty Good Measurement for~$d$ equiprobable Haar-random pure states;
Montanaro~\cite{Montanaro2007} computed $\alpha_2 = 5/6$
and showed $\liminf_{d\to\infty} \alpha_d \geq [8/(3\pi)]^2 \approx 0.72$.

Step 3 (Quantum advantage): For fixed~$d$ and large~$n$,
the classical value scales as
$\omega_c^* = d^2\!/((2d\!-\!1)n) + O(n^{-2})$
(see~\supplM).
Combining with $\E[\omega_q] \approx (d/n)\,\mathcal{L}_1$
from Step~1, the fixed-$d$ asymptotic advantage ratio
$\E[\omega_q]/\omega_c^*$ approaches
$(2d\!-\!1)\,\mathcal{L}_1/d$.
Studying this sequence of ratios as $d \to \infty$,
Montanaro's lim inf bound on~$\alpha_d$ gives
$(2d\!-\!1)\,\mathcal{L}_1/d \geq 2[8/(3\pi)]^4 \approx 1.038 > 1$
in the large-$d$ limit, so quantum advantage holds
for all sufficiently large~$d$.

We remark that the advantage ratio ($\omega_q/\omega_c$) decreases with~$d$ but remains above~$1$
also for every finite $d$ tested ($d = 2, \ldots, 50$), persisting in the limit of infinitely many safe
channels:
\begin{center}
\small
\begin{tabular}{c|ccccc}
$d$ & 2 & 3 & 4 & 5 & $\infty$ \\
\hline
$\omega_q/\omega_c$ & 1.083 & 1.075 & 1.069 & 1.065 & ${\geq}\,1.038$
\end{tabular}
\end{center}

%%% EXPLICIT STRATEGIES %%%
\smallskip
\emph{Harmonic frame:} We now turn to explicit, deterministic seed
vectors that can be defined for all $(n,k)$.
The \emph{harmonic frame} in $\Complex^d$ consists of $n$
vectors from the discrete Fourier transform: $\ket{h_c} =
d^{-1/2}\sum_{j=0}^{d-1} e^{2\pi i\,cj/n}\ket{j}$.  For $d =n-k= 2$,
each seed vector is a qubit state on the equator of the Bloch sphere,
evenly spaced in azimuthal angle, and requires sharing only a single Bell pair.  We obtain:

\textbf{Theorem 3}---\textit{For all $n \geq 3$, the harmonic
strategy satisfies $\omega_q^{\rm harm}(n) > \omega_c^*(n,n\!-\!2)$, with}
\begin{equation}\label{eq:harmonic-ratio}
    \lim_{n \to \infty}
    \frac{\omega_q^{\rm harm}}{\omega_c^*}
    = \frac{3}{4} + \frac{3}{\pi^2} \approx 1.054\,.
\end{equation}
We establish this theorem (see \supplM) by noting that for $d = 2$, each safe set contains two bands, so the Gram matrix
is $2 \times 2$ and the L\"owdin overlaps can be computed in closed
form. The equal spacing of the harmonic frame then reduces the sum
in Eq.~\eqref{eq:omega-q} to a trigonometric series, yielding the
exact ratio in Eq.~\eqref{eq:harmonic-ratio}.

\smallskip
\emph{Simplex frame:} To understand the role of frame geometry, we
study the regular simplex: $d\!+\!1$ equiangular vectors in
$\Complex^d$ with pairwise inner product $|\inp{s_i}{s_j}|^2 =
1/d^2$. This construction gives a game with $n = d + 1$ bands and
only $k = 1$ jammed band.

\textbf{Theorem 4 (Simplex).}---\textit{For the $(d\!+\!1, 1)$-jamming
game with simplex seed vectors,}
\begin{equation}\label{eq:simplex}
    \omega_q^{\mathrm{simplex}} = \frac{1 + (d\!-\!1)\,
    |\inp{v_X^c}{v_Y^c}|^2}{d+1}\,,
\end{equation}
\textit{where $|\inp{v_X^c}{v_Y^c}|^2$ is the squared L\"owdin
overlap between any two safe sets differing in one element.}

Because the simplex is maximally symmetric, all cross-context
overlaps are equal, reducing the computation of $\omega_q$ to a single
parameter that can be computed in closed form (see \supplM).
However, the simplex reveals a \emph{crossover}: it beats the
classical strategy only for $d \leq 5$.  For $d \geq 6$, the safe
sets overlap so heavily ($d$ out of $d+1$ bands) that classical
coordination becomes efficient and the quantum advantage vanishes.
This shows that being a ``good'' frame for some $d$ does not guarantee advantage
for all games.

\emph{Other frames.} We also evaluate three further families of seed
vectors, computing the winning probabilities numerically via
Eq.~\eqref{eq:omega-q}~\cite{JammingNumerics}:
\emph{SIC-POVMs}~\cite{Renes2004}, consisting of $d^2$ equiangular
vectors in~$\Complex^d$ with $|\inp{\phi_a}{\phi_b}|^2 = 1/(d+1)$
for all $a \neq b$ (known to exist in all dimensions up to~53 and
conjectured for all~$d$);
\emph{mutually unbiased bases} (MUBs)~\cite{WoottersFields1989},
$d\!+\!1$ orthonormal bases of~$\Complex^d$ such that
$|\inp{\phi_a}{\phi_b}|^2 = 1/d$ for all vectors from different bases,
providing $d(d\!+\!1)$ seed vectors for prime-power~$d$; and
\emph{Alltop frames}~\cite{Alltop1980}, built from cubic phase
sequences $(\ket{\phi_c})_j = d^{-1/2}\,e^{2\pi i\,cj^3/n}$ with
$|\inp{\phi_a}{\phi_b}|^2 = 1/d$ for $a \neq b$, existing for any
prime $n \geq 5$ and any $d \leq n$.
Figure~\ref{fig:heatmap} summarizes which frame performs best at
each $(n,d)$ point.  MUBs achieve the largest advantage among
explicit constructions ($8.2\%$ at $d = 2$), while the harmonic
frame (available for all $(n,k)$) dominates the interior
of the parameter space.  Upper bounds computed using the SDP
hierarchies~\cite{NPA2008, DLTW2008} (level $1$) for the two smallest games,
$(n,d) = (3,2)$ and $(4,2)$ show that our strategies are
optimal, while for higher values there is a gap.
We remark that some of these constructions (e.g.\ MUBs and AllTop sequences) indeed have their origins in classical spread-spectrum communication~\cite{WinklerChirp1962,SpringerSpread2000,Alltop1980}.

\smallskip
\textit{Discussion:}
One might also wonder whether the explicit strategies examined are
indeed optimal. Numerical optimization over all seed vectors
in~$\Complex^d$ (using the limited-memory BFGS algorithm with box constraints (L-BFGS-B)~\cite{Byrd1995} with multiple random
restarts~\cite{JammingNumerics}) shows that the simplex, SIC, and MUB
frames each attain the optimal seed-based value at their respective
parameter points, as does the harmonic frame for $k = 1$
(Table~\ref{tab:frame-comparison}). For general $(n,d)$, the harmonic
frame is suboptimal but forms an explicit construction available for
all parameter values. 

One might also wonder whether our Ansatz of choosing seed vectors already forms a significant restriction, in the sense that much 
better quantum strategies might otherwise be found. Using numerical optimization
over all rank-1 projective measurements without any seed-vector
structure produces winning probabilities that match the seed-based
optimum to within optimization tolerance for all $(n,d)$ pairs tested (see~\supplM).
This strongly suggests that optimal measurements naturally take the
form of a shared seed frame processed through L\"owdin
orthonormalization.

Our framework of constructing quantum strategies
from seed vectors via L\"owdin orthonormalization opens the door to study general
quantum-enhanced spectrum allocation, and may prove useful for other
coordination games where distributed parties must make compatible
decisions from local observations. 

Since the $d = 2$ harmonic strategy
requires only a single shared Bell pair, a proof-of-concept is
within reach of current quantum network
hardware~\cite{Stolk2024, Krutyanskiy2023}. It would be interesting
to explore whether entanglement-enhanced coordination can benefit
real-world use cases such as anti-jamming frequency
hopping~\cite{Popper2010}, ultra-reliable low-latency communication
in next-generation wireless networks~\cite{Bennis2018}, or
distributed frequency management in satellite
constellations~\cite{LEOsurvey2025}.

\acknowledgments SW thanks Jeroen Grimbergen and Sounak Kar for useful comments on an earlier version of this manuscript. 
SW was supported by an NWO VICI grant. Claude Code (Opus 4.6 and 4.7) was used to assist this research including the software used for the numerical computations, plotting,
unifying editing, converting handrawn ideas to figures, and numerical exploration of conjectures during the research process. Gemini was used to search the literature. Aristotle (harmonic.fun) and Claude Code (Opus 4.7) were used in converting latex proofs to Lean 4.

\bibliography{jamming}

%%%%%%%%%%%%%%%%%%%%%%%%%%%%%%%%%%%%%%%%%%%%%%%%%%%%%%%%%%%%%%%%%%%%%
%%% supplMEMENTARY MATERIAL (single-column appendix)
%%%%%%%%%%%%%%%%%%%%%%%%%%%%%%%%%%%%%%%%%%%%%%%%%%%%%%%%%%%%%%%%%%%%%

\onecolumngrid
\appendix
\setcounter{secnumdepth}{3}

%\vspace{1em}
%\begin{center}
%\rule{\textwidth}{0.5pt}\\[1em]
%{\Large\textbf{supplementary  Material}}\\[0.5em]
%{\large Quantum Advantage for Coordinated Frequency Selection Against Distributed Jammers}\\[0.5em]
%Stephanie Wehner\\[1em]
%\rule{\textwidth}{0.5pt}
%\end{center}
%\vspace{1em}

In this Appendix, we provide the full technical details underpinning the results presented in the main text. 

\input{classicalBound}
\input{generalQuantum}

\input{randomFrames}

\input{exampleStrategies}

\end{document}

%% file: classicalBound.tex
\section{\texorpdfstring{Classical Value of $(n,k)$-Jamming Games}{Classical Value of (n,k)-Jamming Games}}

In this section, we present the detailed derivation of the maximum probability $\omega_c(n,k)$ to win the $(n,k)$-jamming game using any classical strategy. In keeping with the established literature~\cite{Brunner2014} we also call $\omega_{cl}(n,k)$ the \emph{classical value} of the game. Note that when studying $\omega_{cl}$ we can take $f(x) \in x$ and $g(y) \in y$, 
since otherwise Alice and Bob would return a losing answer. We will refer to a strategy of this form as \emph{valid}.
 
\begin{theorem}\label{thm:classical-value}
The classical value of the $(n,k)$-jamming game is
\begin{align}\label{eq:classical-value}
\omega_c(n,k) = \frac{1}{\binom{n}{k}^2} \sum_{i=0}^{k} \binom{n-1-i}{k-i}^2.
\end{align}
\end{theorem}
Our proof proceeds in three steps:
\begin{enumerate}
    \item First, we show that optimal strategies have the form $(f,f)$, where Alice and Bob use the same deterministic function $f$ to produce their answers. 
    \item Second, we show that any optimal aligned strategy must \emph{saturate} certain constraints, assigning the maximum possible number of safe sets to the smallest channel.
    \item Third, we construct an explicit greedy strategy and prove its optimality by induction.
\end{enumerate}

\subsection{Aligned Strategies are Optimal}
A deterministic strategy consists of functions $f: \mathcal{S} \to [n]$ for Alice, and $g: \mathcal{S}\to [n]$ for Bob, where the optimal classical value of a game is attained by such deterministic functions.  We now first show that choosing $g=f$ attains the optimal classical value. More formally, 
\begin{lemma}\label{lem:aligned}
Among all deterministic strategy pairs $(f, g)$, the optimal winning probability is achieved by an aligned pair $(f, f)$.
\end{lemma}
We establish this lemma in two steps. First, we show that the structure of the jamming game implies that the number of wins takes a simple product form. Second, we apply the Cauchy--Schwarz inequality to conclude that aligned strategies, i.e. $g=f$, are optimal.

\begin{proposition}\label{prop:win-simplifies}
The $(n,k)$-jamming game is synchronous: a valid strategy pair $(f, g)$ wins on input $(x, y)$ if and only if $f(x) = g(y)$. Consequently, the number of winning input pairs is
\begin{align}\label{eq:win-product}
W(f, g) := |\{(x, y) \in \mathcal{S} \times \mathcal{S} : (f,g) \text{ wins on } (x,y)\}| = \sum_{c \in [n]} |f^{-1}(c)| \cdot |g^{-1}(c)|.
\end{align}
\end{proposition}
\begin{proof}
Note that the winning condition requires
\begin{enumerate}
    \item[(i)] $f(x) = g(y)$ (Alice and Bob output the same channel), and
    \item[(ii)] $f(x) \in x \cap y$ (the common output is safe at both sites).
\end{enumerate}

We first claim that (ii) follows from (i) together with validity. Indeed, suppose $f(x) = g(y) = c$. By validity of $f$, we have $c = f(x) \in x$. By validity of $g$, we have $c = g(y) \in y$. Hence $c \in x \cap y$, establishing (ii).
Therefore, the winning condition simplifies to $f(x) = g(y)$. The number of winning pairs is thus
\begin{align}
W(f, g) = |\{(x, y) : f(x) = g(y)\}| = \sum_{c \in [n]} |f^{-1}(c)| \cdot |g^{-1}(c)|,
\end{align}
where the second equality follows by partitioning according to the common output value $c$.
\end{proof}

\begin{proposition}\label{prop:aligned-optimal}
Among all valid strategy pairs $(f,g)$, the number of wins $W(f,g)$ is maximized by an aligned pair $(f, f)$.
\end{proposition}

\begin{proof}
Let $a_c = |f^{-1}(c)|$ and $b_c = |g^{-1}(c)|$. By Eq.~\eqref{eq:win-product}, $W(f,g) = \sum_{c} a_c b_c$.
Both vectors $(a_c)_{c \in [n]}$ and $(b_c)_{c \in [n]}$ lie in $\mathbb{R}_{\geq 0}^n$ and satisfy $\sum_c a_c = \sum_c b_c = |\mathcal{S}|$. By the Cauchy--Schwarz inequality,
\begin{align}
\sum_c a_c b_c \leq \sqrt{\sum_c a_c^2} \cdot \sqrt{\sum_c b_c^2},
\end{align}
with equality if and only if $(a_c)$ and $(b_c)$ are proportional. Since both vectors have the same $\ell^1$-norm $|\mathcal{S}|$, proportionality implies $a_c = b_c$ for all $c$.

Therefore, for any feasible allocation $(a_c)$ achievable by some strategy $f$, the number of wins is maximized when $g$ has the same allocation. The choice $g = f$ achieves this, giving $W(f, f) = \sum_c a_c^2$.
\end{proof}

Lemma~\ref{lem:aligned} follows immediately from Propositions~\ref{prop:win-simplifies} and~\ref{prop:aligned-optimal}.

Note that the product structure in Eq.~\eqref{eq:win-product} is special to the jamming game. In a general nonlocal game, the winning condition may depend on the input pair $(x, y)$ in ways that do not factor through output equality. For such games, aligned strategies may be suboptimal.

\subsection{Saturation at Optimum}

By Lemma~\ref{lem:aligned}, finding the optimal classical strategy reduces to maximizing $W(f,f) = \sum_c |f^{-1}(c)|^2$ over valid strategies $f$. We now establish a key property of optimal strategies.
For each channel $c \in [n]$, define the \emph{star} at $c$ as
\begin{align}\label{eq:star-def}
X_c := \{x \in \mathcal{S} : c \in x\},
\end{align}
the collection of safe sets containing channel $c$. Since a safe set is an $(n-k)$-subset of $[n]$, and requiring $c \in x$ leaves $n-k-1$ elements to be chosen from the remaining $n-1$ channels, we have
\begin{align}\label{eq:star-size}
|X_c| = \binom{n-1}{n-k-1} = \binom{n-1}{k}.
\end{align}
Note that this quantity is independent of $c$, reflecting the symmetry of the problem under channel permutations.

A valid strategy can only assign $f(x) = c$ a safe set $x$ to a channel $c$ if $c \in x$. Hence
\begin{align}\label{eq:capacity-constraint}
|f^{-1}(c)| \leq |X_c| = \binom{n-1}{k} \quad \text{for all } c \in [n].
\end{align}
We call this the \emph{capacity constraint} for channel $c$.

We now establish the following saturation lemma. Intuitively, it follows from the 
principle that convex maximization rewards inequality:
Concentrating mass on a single channel (up to its capacity) always increases the sum of squares $\sum_c |f^{-1}(c)|^2$, regardless of how the remaining mass is distributed.

\begin{lemma}[Saturation]\label{lem:saturation}
Let $f$ be an optimal strategy, and let $d_c = |f^{-1}(c)|$ be the induced allocation. Then
\begin{align}
\max_{c \in [n]} d_c = \binom{n-1}{k}.
\end{align}
That is, some channel must be assigned exactly all safe sets containing it.
\end{lemma}

\begin{proof}
Since the input distribution and winning condition are invariant under permutation of channel labels (each star has the same size $\binom{n-1}{k}$ by Eq.~\eqref{eq:star-size}), we may relabel channels so that the allocation is sorted in decreasing order: $d_1 \geq d_2 \geq \cdots \geq d_n$. We show that $d_1 = \binom{n-1}{k}$.

Suppose by contradiction that $d_1 < \binom{n-1}{k} = |X_1|$. Then there exists a safe set $x \in X_1$ (i.e., $1 \in x$) that is not assigned to channel~1, i.e. $f(x) \neq 1$. Let $c = f(x) \geq 2$ be its current assignment.

Consider the modified strategy $f'$ that reassigns $x$ from channel $c$ to channel $1$:
\begin{align}
f'(y) = \begin{cases} 1 & \text{if } y = x, \\ f(y) & \text{otherwise.} \end{cases}
\end{align}
This is valid since $1 \in x$. The new allocation satisfies $d_1' = d_1 + 1$, $d_c' = d_c - 1$, and $d_{c'}' = d_{c'}$ for $c' \notin \{1, c\}$.

The change in objective is:
\begin{align}
W(f', f') - W(f, f) &= (d_1 + 1)^2 + (d_c - 1)^2 - d_1^2 - d_c^2 \\
&= 2d_1 + 1 - 2d_c + 1 \\
&= 2(d_1 - d_c + 1).
\end{align}
Since $d_1 \geq d_c$ by our sorting assumption, we have $d_1 - d_c + 1 \geq 1$, and thus
\begin{align}
W(f', f') - W(f, f) \geq 2 > 0.
\end{align}
This contradicts the optimality of $f$. Therefore $d_1 = \binom{n-1}{k}$.
\end{proof}

\subsection{The Greedy Strategy}

We now construct an explicit optimal strategy and prove its optimality. 
\begin{definition}[Greedy Strategy]\label{def:greedy}
The \emph{greedy strategy} $f^*: \mathcal{S} \to [n]$ assigns each safe set $x$ to the smallest-numbered channel it contains:
\begin{align}
f^*(x) = \min(x).
\end{align}
\end{definition}
We call this the greedy strategy due to its equivalent recursive description:
\begin{enumerate}
    \item Assign all safe sets containing channel $1$ to channel $1$: $f^{-1}(1) = X_1$.
    \item The remaining safe sets $\mathcal{S} \setminus X_1$ are exactly the $(n-k)$-subsets of $\{2, \ldots, n\}$, which form the input space for the $(n-1, k-1)$-jamming game on channels $\{2, \ldots, n\}$.
    \item Recurse: apply the greedy strategy to $\mathcal{S} \setminus X_1$.
\end{enumerate}
The greedy strategy has a natural interpretation: prioritize channels in order $1, 2, \ldots, n$, assigning each safe set to the smallest-numbered channel it contains. This concentrates as many safe sets as possible onto a single channel before moving to the next, producing an allocation that is maximally unequal across channels. Since the sum of squares $\sum_c m_c^2$ is a convex function, it is maximized precisely by such unequal allocations (made rigorous by the theory of majorization and Schur-convex functions~\cite{Bhatia1997}).

\begin{proposition}[Greedy Allocation]\label{prop:greedy-allocation}
The greedy strategy achieves the allocation
\begin{align}
d_c^* = |{f^*}^{-1}(c)| = \begin{cases} \binom{n-c}{k-c+1} & \text{if } 1 \leq c \leq k+1, \\ 0 & \text{if } c > k+1. \end{cases}
\end{align}
\end{proposition}

\begin{proof}
By induction on $k$. For $k = 0$, the unique safe set $[n]$ is assigned to channel $1 = \min([n])$, giving $d_1^* = 1 = \binom{n-1}{0}$ and $d_c^* = 0$ for $c \geq 2$.
For $k \geq 1$, the greedy strategy assigns $d_1^* = |X_1| = \binom{n-1}{k}$ to channel $1$. By the recursive structure, the remaining allocation $(d_2^*, \ldots, d_n^*)$ equals the greedy allocation for the $(n-1, k-1)$-game. By the inductive hypothesis:
\begin{align}
d_c^* = \binom{(n-1)-(c-1)}{(k-1)-(c-1)+1} = \binom{n-c}{k-c+1}
\end{align}
for $2 \leq c \leq k+1$, and $d_c^* = 0$ for $c > k+1$.
\end{proof}

We now evaluate $W$ for the greedy strategy.
\begin{proposition}[Greedy Wins]\label{prop:greedy-wins}
The greedy strategy achieves
\begin{align}\label{eq:greedy-wins}
W(f^*, f^*) = \sum_{i=0}^{k} \binom{n-1-i}{k-i}^2.
\end{align}
\end{proposition}

\begin{proof}
Let $W(n,k)$ denote the number of winning pairs for the greedy strategy on the $(n,k)$-game. The recursive structure gives:
\begin{align}\label{eq:W-recursion}
W(n,k) = |X_1|^2 + W(n-1, k-1) = \binom{n-1}{k}^2 + W(n-1, k-1),
\end{align}
with base case $W(n, 0) = 1$ (the single safe set $[n]$ paired with itself).
We prove Eq.~\eqref{eq:greedy-wins} by induction on $k$. 

\emph{Base case} ($k = 0$): $W(n,0) = 1 = \binom{n-1}{0}^2$. 

\emph{Inductive step}: Assume the formula holds for $k-1$. Then:
\begin{align}
W(n,k) &= \binom{n-1}{k}^2 + W(n-1, k-1) && \text{(by Eq.~\eqref{eq:W-recursion})} \\
&= \binom{n-1}{k}^2 + \sum_{i=0}^{k-1} \binom{n-2-i}{k-1-i}^2 && \text{(by inductive hypothesis)} \\
&= \binom{n-1}{k}^2 + \sum_{j=1}^{k} \binom{n-1-j}{k-j}^2 && \text{(substituting $j = i+1$)} \\
&= \sum_{i=0}^{k} \binom{n-1-i}{k-i}^2. && \qedhere
\end{align}
\end{proof}

\subsection{Optimality of the Greedy Strategy}

We now prove that the greedy strategy is optimal.
\begin{proposition}[Greedy Optimality]\label{prop:greedy-optimal}
Among all valid strategies $f: \mathcal{S} \to [n]$, the greedy strategy maximizes $W(f,f) = \sum_c |f^{-1}(c)|^2$.
\end{proposition}

\begin{proof}
We proceed by strong induction on $k$.

\emph{Base case} ($k = 0$): The unique safe set $[n]$ must be assigned to some channel. Any assignment yields $W = 1$, matching the greedy.

\emph{Inductive step}: Assume the result holds for all $(n', k')$ with $0 \leq k' < k$. Let $f$ be any optimal strategy for the $(n,k)$-game.
By Lemma~\ref{lem:saturation}, some channel achieves the maximum allocation $\binom{n-1}{k}$. By $S_n$-symmetry, we may relabel channels so that $d_1 = \binom{n-1}{k}$.
Since $|f^{-1}(1)| = |X_1|$ and $f^{-1}(1) \subseteq X_1$ (by validity), we have $f^{-1}(1) = X_1$. That is, channel~1 receives exactly all safe sets containing it.
The remaining safe sets $\mathcal{S} \setminus X_1$ must be assigned to channels $\{2, \ldots, n\}$. Observe that:
\begin{itemize}
    \item $\mathcal{S} \setminus X_1$ consists of all $(n-k)$-subsets of $[n]$ that do not contain channel~1.
    \item These are precisely the $(n-k)$-subsets of $\{2, \ldots, n\}$.
    \item This is exactly the input space for the $(n-1, k-1)$-jamming game on channels $\{2, \ldots, n\}$.
\end{itemize}
The restriction $f|_{\mathcal{S} \setminus X_1}$ is a valid strategy for this $(n-1, k-1)$-subgame: for any $x \in \mathcal{S} \setminus X_1$, we have $1 \notin x$, hence $f(x) \neq 1$ (since $f^{-1}(1) = X_1$ and $x \notin X_1$), and $f(x) \in x$ by validity of the original strategy, so $f(x)$ assigns $x$ to a channel in $x \cap \{2, \ldots, n\}$. By the inductive hypothesis:
\begin{align}
\sum_{c=2}^{n} |f^{-1}(c)|^2 \leq W_{\mathrm{greedy}}(n-1, k-1).
\end{align}
Therefore, we have
\begin{align}
W(f,f) &= |f^{-1}(1)|^2 + \sum_{c=2}^{n} |f^{-1}(c)|^2 \\
&\leq \binom{n-1}{k}^2 + W_{\mathrm{greedy}}(n-1, k-1) \\
&= W_{\mathrm{greedy}}(n, k),
\end{align}
where the last equality is the recursion~\eqref{eq:W-recursion}.
Since the greedy strategy achieves this upper bound, it is optimal.
\end{proof}

\subsection{Proof of Theorem ~\ref{thm:classical-value}}

We are now ready to put all ingredients together to prove the claimed statement on the maximum classical winning probability. 

\begin{proof}[Proof of Theorem ~\ref{thm:classical-value}]
Combining Lemma~\ref{lem:aligned}, Proposition~\ref{prop:greedy-wins}, and Proposition~\ref{prop:greedy-optimal}:
\begin{align}
\omega_c(n,k) = \frac{W_{\max}}{|\mathcal{S}|^2} = \frac{W_{\mathrm{greedy}}(n,k)}{\binom{n}{k}^2} = \frac{1}{\binom{n}{k}^2} \sum_{i=0}^{k} \binom{n-1-i}{k-i}^2. 
\end{align}
\end{proof}

\begin{proposition}[Classical scaling for fixed $d$]\label{prop:classical-scaling}
For fixed $d \geq 2$ and $k = n - d$,
\begin{align}\label{eq:classical-scaling}
\omega_c(n, n-d) = \frac{d^2}{(2d-1)\,n} + O(n^{-2}).
\end{align}
\end{proposition}

\begin{proof}
When $d$ is fixed and $n \to \infty$, safe sets become sparse $d$-subsets of $[n]$. Substituting $j = n - 1 - i$ in the sum from Theorem~\ref{thm:classical-value} and using $\binom{n-1-i}{k-i} = \binom{n-1-i}{d-1}$ gives
\begin{align}\label{eq:omega-c-reindex}
\omega_c = \frac{1}{\binom{n}{d}^2}\sum_{j=d-1}^{n-1}\binom{j}{d-1}^2.
\end{align}
For the numerator, we expand the binomial coefficient for fixed $d$ and large $j$ as
\begin{align}
\binom{j}{d-1} = \frac{j^{d-1}}{(d-1)!}\!\left(1 + O(1/j)\right)\!,
\end{align}
	so that $\binom{j}{d-1}^2 = j^{2(d-1)}/((d-1)!)^2 \cdot (1 + O(1/j))$. Summing over $j$ and using the power sum formula $\sum_{j=1}^{m} j^p = m^{p+1}/(p+1) + O(m^p)$~\cite{powerSum} we obtain
\begin{align}\label{eq:sum-binom-squared}
\sum_{j=d-1}^{n-1}\binom{j}{d-1}^2 &= \frac{1}{((d-1)!)^2}\!\left(\sum_{j=1}^{n-1} j^{2(d-1)} + O\!\left(\sum_{j=1}^{n-1} j^{2d-3}\right)\right) = \frac{n^{2d-1}}{((d-1)!)^2(2d-1)} + O(n^{2d-2}),
\end{align}
where the $O(1/j)$ correction contributes $O(n^{2d-2})$ via $\sum j^{2d-3} = O(n^{2d-2})$. Note that the expansion $\binom{j}{d-1} = j^{d-1}/(d-1)! \cdot (1 + O(1/j))$ is applied only for $j \geq d-1$; the terms $j = 1, \ldots, d-2$ satisfy $\binom{j}{d-1} = 0$ and so contribute nothing to the left-hand side. Since $d$ is fixed, the finitely many small-$j$ terms contribute $O(1)$, which is absorbed into $O(n^{2d-2})$. For the denominator, $\binom{n}{d} = n^d/(d!) \cdot (1 + O(1/n))$, so $\binom{n}{d}^2 = n^{2d}/(d!)^2 \cdot (1 + O(1/n))$. Combining:
\begin{align}
\omega_c = \frac{(d!)^2}{n^{2d}}\!\left(1 + O(1/n)\right) \cdot \frac{n^{2d-1}}{((d-1)!)^2(2d-1)}\!\left(1 + O(1/n)\right) = \frac{d^2}{(2d-1)n} + O(n^{-2}).
\end{align}
The $1/n$ scaling reflects that coordination becomes difficult when safe sets are sparse: the number of safe sets is $|\mathcal{S}| = \binom{n}{d} \sim n^d/d!$, but the number of pairs sharing the same greedy assignment grows only as $n^{2d-1}$.
\end{proof}

\begin{remark}[Asymptotic Behavior of Classical Value]\label{rem:classical-scaling}
We note the asymptotic behaviour of the classical value in two further regimes. Throughout, we write $d = n-k$ for the safe set size, and define the normalized terms
\begin{align}\label{eq:ai-def}
a_i := \frac{\binom{n-1-i}{k-i}}{\binom{n}{k}}, \qquad i = 0, 1, \ldots, k,
\end{align}
so that $\omega_c(n,k) = \sum_{i=0}^{k} a_i^2$. Note that $a_0 = d/n$ (since $\binom{n-1}{k}/\binom{n}{k} = (n-k)/n$), and the ratio of consecutive terms satisfies
\begin{align}\label{eq:ratio-consecutive}
\frac{a_i}{a_{i-1}} = \frac{k - i + 1}{n - i}, \qquad i = 1, \ldots, k.
\end{align}

\textbf{(i) Few jammers (fixed $k$, large $n$):}
When $k$ is fixed and $n \to \infty$, the sum has $k+1$ terms. Each ratio in Eq.~\eqref{eq:ratio-consecutive} satisfies $(k-i+1)/(n-i) = (k-i+1)/n \cdot (1 + O(1/n))$, so
\begin{align}
a_i = \frac{d}{n} \cdot \frac{k^{\underline{i}}}{n^i}\bigl(1 + O(1/n)\bigr),
\end{align}
where $k^{\underline{i}} = k(k-1)\cdots(k-i+1)$ is the falling factorial. Since $k$ is fixed, $a_i^2 = O(n^{-2i-2})$ for $i \geq 1$, and the sum is dominated by the first two terms
\begin{align}
\omega_c = \frac{d^2}{n^2} + \frac{k^2 d^2}{n^4} + O(n^{-4}) = \frac{d^2}{n^2}\!\left(1 + \frac{k^2}{n^2} + O(n^{-4})\right)\!.
\end{align}
Comparing with the geometric series $(d/n)^2 \sum_{i=0}^{\infty} (k/n)^{2i} = (d/n)^2/(1 - k^2/n^2) = d^2/(n^2 - k^2) = (n-k)/(n+k)$, the error from truncating at $i = k$ (fixed) and from the $O(1/n)$ corrections to each ratio are both $O(n^{-2})$, giving
\begin{align}
    \omega_c(n,k) = \frac{n-k}{n+k} + O(n^{-2}) = \frac{d}{2n-d} + O(n^{-2}).
\end{align}
As a check: for $k=1$, this gives $\omega_c \approx (n-1)/(n+1)$, which agrees with the exact value $((n-1)^2 + 1)/n^2$ to order $O(n^{-2})$.

\textbf{(ii) Proportional scaling ($k = \alpha n$ for fixed $\alpha \in (0,1/2)$):}
When the jamming fraction $\alpha = k/n$ is held constant, both $k$ and $d$ grow linearly with $n$. Approximating each ratio in Eq.~\eqref{eq:ratio-consecutive} by $\alpha$ gives the geometric series
\begin{align}
\omega_c \approx \frac{d^2}{n^2}\sum_{i=0}^{k}\alpha^{2i} = \frac{(1-\alpha)^2}{1-\alpha^2} = \frac{1-\alpha}{1+\alpha},
\end{align}
since $\alpha^{2k} = \alpha^{2\alpha n}$ is exponentially small. To quantify the error: the true ratio $(k-i+1)/(n-i)$ differs from $\alpha = k/n$ by $O(1/n)$ for each $i$, and these corrections accumulate over the $O(1)$ effective terms (weighted by $\alpha^{2i}$), producing a net $O(1/n)$ correction. Hence
\begin{align}
    \omega_c(n,k) = \frac{1-\alpha}{1+\alpha} + O(n^{-1}).
\end{align}
The weaker error term compared to~(i) and Proposition~\ref{prop:classical-scaling} reflects the fact that $k$ now grows with~$n$. This leading term is a constant independent of $n$, interpolating smoothly between $\omega_c \to 1$ as $\alpha \to 0$ (no jamming) and $\omega_c \to 1/3$ as $\alpha \to 1/2$ (half the channels jammed). For example, when $20\%$ of channels are jammed ($\alpha = 0.2$), the classical value is approximately $2/3$.
\end{remark}
\begin{remark}[Graph-Theoretic Interpretation]
Note that the problem of finding the optimal strategy also admits a graph-theoretic interpretation. Define the \emph{compatibility graph} $G$ with vertex set $\mathcal{S}$ and an edge between $x$ and $y$ whenever $|x \cap y| \geq 1$. For each channel $c$, the star $X_c$ forms a clique in $G$: any two safe sets containing $c$ share at least the element $c$.

A valid strategy partitions $\mathcal{S}$ such that each part $f^{-1}(c)$ is contained in the star $X_c$. The number of wins is $\sum_c |f^{-1}(c)|^2$. The greedy strategy solves this by iteratively assigning all remaining vertices of a star to that channel, exploiting the convexity of $x \mapsto x^2$.
\end{remark}

%% file: generalQuantum.tex
\section{A framework for quantum strategies}
We now introduce a general framework for constructing quantum strategies for $(n,k)$-jamming games, and
establish some general properties such strategies obey. In Section~\ref{sec:examples}, we then use our framework
to provide explicit quantum strategies that provide a quantum advantage in $(n,k)$-jamming games. 

\subsection{State and Measurements}

Recall that a quantum strategy consists of a choice of shared state $\ket{\Psi}$ and measurements for Alice and Bob. It is known that for synchronous games such as the jamming game, in which the players must produce the same output $a=b$ for the same input $x=y$, the optimal quantum strategy can always be taken to use the maximally entangled state~\cite{PaulsenSeverini2016} with a minimum dimension equal to the number of possible outputs per input, i.e., the safe set size $d = n-k$. We thus 
fix the state $\ket{\Psi}$ that Alice and Bob will use to be the maximally entangled state
\begin{align}
\ket{\Psi_d} = \frac{1}{\sqrt{d}} \sum_{j=0}^{d-1} \ket{j}\ket{j}\ ,
\end{align}
where we will choose the minimal possible dimension, that is, $d=n-k$ is equal to the size of the safe sets. This choice is motivated by the desire to provide a practical quantum advantage that can be realized using the smallest possible quantum state.
We now let Alice's measurements be given by the bases
\begin{align}
\mathcal{A}_x = \big\{\ket{v_x^c} \mid c \in x\ \}, 
\end{align}
where $x \in \mathcal{S}$ denote the safe set leading to the choice of measurement basis, and the measurement outcome $c$ denote the choice of frequency band that the measuring party will make. We also define $\ket{v_x^c} = 0$ for all $c \notin x$ to be the zero vector. 
Bob uses the same bases as Alice up to conjugation
\begin{align}
\mathcal{B}_y = \big\{\ket{w_y^c} \mid c \in y\ \}, 
\end{align}
where $\ket{w_y^c} = \ket{v_y^c}^*$.
For the maximally entangled state $\ket{\Psi_d}$ 
the joint probability of Alice obtaining outcome $a$ and Bob obtaining outcome $b$ is given by
\begin{align}
P(a,b \mid x,y) = \bra{\Psi_d} \outp{v_x^a}{v_x^a} \otimes \outp{w_y^b}{w_y^b} \ket{\Psi_d} = \frac{1}{d} \left| \inp{v_x^a}{v_y^b} \right|^2\ , 
\end{align}
where we used the identity $\bra{\Psi_d} A \otimes B\ket{\Psi_d} = \tr(AB^T)/d$~\cite{NielsenChuang2010}.

\subsubsection{Developing intuition}
We now first rewrite the winning probability in a way that will make it more clear 
how we can choose measurements that lead to a large quantum winning probability. 
For the jamming game, Alice and Bob win if and only if they output the same safe channel $c \in x \cap y$. 
For a strategy of the above form, we have
\begin{align}\label{eq:omega_q_general}
\omega_q & = \frac{1}{|S|^2} \sum_{x,y \in S} \sum_{\substack{c \in x \cap y}} P(c,c\mid x,y)\ ,\\
& = \frac{1}{|S|^2} \sum_{x,y \in S} \sum_{\substack{c \in x \cap y}} \frac{1}{d} 
\left| \inp{v_x^c}{v_y^c} \right|^2\ ,\\
&=\frac{1}{d |S|^2} \sum_{c \in [n]} \sum_{\substack{x\\c \in x}} \sum_{\substack{y\\c\in y}}
\left| \inp{v_x^c}{v_y^c} \right|^2\ ,\label{eq:exchanged}\\
&=\frac{1}{d |S|^2} \sum_{c \in [n]} \tr\left[\left(\sum_{\substack{x\\c \in x}} \outp{v_x^c}{v_x^c}\right) \left(\sum_{\substack{y\\c\in y}} \outp{v_y^c}{v_y^c}\right)\right]
\label{eq:traceDef}
\\
&=\frac{1}{d |S|^2} \sum_{c \in [n]} \tr\left[A_c^2\right]\ , \label{eq:inTermsOfA}
\end{align}
with
\begin{align}
A_c := \sum_{\substack{x\\ c \in x}} \outp{v_x^c}{v_x^c}\ . 
\end{align}
Eq.~\eqref{eq:exchanged} follows from the fact that $\ket{v_s^c} = 0$ for $c \notin s$ allowing us to exchange the order of the sums. Eq.~\eqref{eq:traceDef} follows from the fact that the trace operation is linear and cyclic. 

Let us now use the rewritten expression to develop some intuition for how to choose the measurement bases $\mathcal{A}_s$ to maximize $\omega_q$.
From Eq.~\eqref{eq:inTermsOfA}
we see that our goal is to make $\tr[A_c^2]$ as large as possible for each channel $c$. The operator $A_c = \sum_{x \ni c} \outp{v_x^c}{v_x^c}$ is a sum of $N_c := \binom{n-1}{d-1}$ rank-one projectors, one for each safe set containing channel $c$.
To understand what makes $\tr[A_c^2]$ large, note that we have by definition
\begin{align}\label{eq:Aexpand}
\tr[A_c^2] = \sum_{x,y \ni c} \left|\inp{v_x^c}{v_y^c}\right|^2 = N_c + \sum_{\substack{x,y \ni c \\ x \neq y}} \left|\inp{v_x^c}{v_y^c}\right|^2\ .
\end{align}
The diagonal terms ($x = y$) contribute $N_c$ since each $\ket{v_x^c}$ is a unit vector. The off-diagonal terms ($x \neq y$) contribute positively whenever the vectors $\ket{v_x^c}$ and $\ket{v_y^c}$ are not orthogonal.

We thus see that we want the measurement vectors $\ket{v_x^x}$ and $\ket{v_y^x}$ for the same channel $c$ across different safe sets $x$ and $y$ to have significant overlap. If all vectors $\ket{v_x^c}$ for fixed $c$ were identical, we would have $|\inp{v_x^c}{v_y^c}|^2 = 1$ for all pairs, maximizing the sum. However, we face a constraint: for each safe set $x$, the vectors $\{\ket{v_x^c}\}_{c \in x}$ must form an orthonormal basis of $\mathbb{C}^d$.
The challenge is thus to construct orthonormal bases that are as ``aligned'' as possible across different safe sets, while respecting the orthonormality constraint within each safe set. 

\subsubsection{Framework}
Given the focus on vectors for channels $c$, our Ansatz begins with a single family of ``seed'' vectors $\{\ket{\phi_c}\}_{c \in [n]} \in \Complex^{d}$ with $d=n-k$, one for each channel $c$. The intuition is that by focusing on channel-specific vectors, we can encourage alignment across different safe sets for the same channel. Of course, we cannot measure directly in terms of these seed vectors for two reasons. First, in a synchronous game Alice and Bob benefit from producing the same outcome when performing the same measurement, meaning it is convenient to measure the maximally entangled state in a basis (or projective measurement), rather than some arbitrary measurement. Second, there are $n > d$ seed vectors, too many to form a measurement basis in $\Complex^d$.

To construct valid measurements, we could then orthonormalize the seed vectors within each safe set. That is, for a safe set $x$, we orthonormalize the vectors $\{\ket{\phi_c}\}_{c \in x}$ to obtain our measurement basis. The core of our framework thus proceeds in two steps:
\begin{enumerate}
    \item[(i)] \textbf{Seed vectors:} Define a family of $n$ vectors $\{\ket{\phi_c}\}_{c \in[n]}$ in $\mathbb{C}^d$ respecting the cyclic symmetry of the channel labels.
    \item[(ii)] \textbf{Orthonormalization:} For each safe set $x$, orthonormalize the seed vectors $\{\ket{\phi_c}\}_{c \in x}$ to obtain the measurement basis $\{\ket{v_x^c}\}_{c \in x}$.
\end{enumerate}

For the orthonormalization step, we employ throughout the L\"owdin symmetric orthonormalization procedure~\cite{Lowdin1950}, whenever the seed vectors are all linearly independent. This choice is inspired by quantum chemistry, where one seeks an orthonormal basis that maximally retains the character of atomic orbital wavefunctions. The L\"owdin procedure produces orthonormal vectors that are ``closest'' to the original vectors in a least-squares sense~\cite{Lowdin1950,Mayer2003}. This property is intuitively well-suited to our application, since the orthonormalized vectors $\ket{v_x^c}$ retain as much structure of the seed vectors $\ket{\phi_c}$ as possible, helping preserve overlaps between vectors associated with the same channel across different safe sets.

What if the seed vectors are not linearly independent? It turns out that in this case, we can still follow the L\"owdin intuition but with a twist (see below): we obtain the measurements by the so-called \emph{Pretty Good Measurement}~\cite{Belavkin1975,HausladenWootters1994}, sometimes also called the \emph{square-root measurement}~\cite{HausladenJozsaWootters1996,Eldar2003}, which is equivalent to L\"owdin orthonormalization whenever the seed vectors are linearly independent.

Note that our framework is fully general, and can be instantiated for different types of seed vectors. In Section~\ref{sec:exampleStrategies}, we consider a number of different types of seed vectors and evaluate the quantum advantage in each case.

\subsection{Orthogonalized measurements and their properties}

We now explain in detail how we obtain the measurement vectors using a general form of L\"owdin orthonormalization for general seed vectors. This general form is also known as the pretty good measurement in quantum information~\cite{Belavkin1975,HausladenWootters1994}. We then establish some general properties of the resulting set of measurements.

\subsubsection{\texorpdfstring{The Pretty Good Measurement (PGM) and L\"owdin orthonormalization}{The Pretty Good Measurement and Lowdin orthonormalization}}\label{sec:lowdin}\label{sec:gram}

Let $\{\ket{\phi_c}\}_{c \in [n]}$ be any collection of $n$ unit vectors in $\Complex^d$ with $d \leq n$. Such a collection is called a \emph{frame} for $\Complex^d$~\cite{Christensen2016,Waldron2018}. 
For any subset $\mathcal{X} \subseteq [n]$ with $|\mathcal{X}| = d$, we define the \emph{Gram matrix}~\cite{Gram1883,Bhatia2007} $G_\mathcal{X}$ as the $d \times d$ Hermitian matrix with entries indexed by channels $c, c' \in \mathcal{X}$:
\begin{align}\label{eq:gram-definition}
[G_\mathcal{X}]_{c,c'} := \inp{\phi_c}{\phi_{c'}}.
\end{align}
The Gram matrix $G_{\mathcal{X}}$ encodes the inner products among the seed vectors restricted to the safe set $\mathcal{X}$. It is always positive semidefinite, and is positive definite if and only if the seed vectors $\{\ket{\phi_c}\}_{c \in \mathcal{X}}$ are linearly independent.

Let $G_\mathcal{X}^+$ denote the Moore--Penrose pseudoinverse of $G_\mathcal{X}$, and let $(G_\mathcal{X}^+)^{1/2}$ denote its positive semidefinite square root.

\begin{definition}[PGM / Generalized L\"owdin orthonormalization]\label{def:lowdin}
For a safe set $\mathcal{X}$ with $|\mathcal{X}| = d$, the \emph{measurement vectors} are
\begin{align}\label{eq:lowdin-def}
\ket{v_\mathcal{X}^c} := \sum_{c' \in \mathcal{X}} [(G_\mathcal{X}^+)^{1/2}]_{c',c} \ket{\phi_{c'}}, \quad c \in \mathcal{X}.
\end{align}
The corresponding POVM elements are $M_\mathcal{X}^c := \outp{v_\mathcal{X}^c}{v_\mathcal{X}^c}$.
\end{definition}

This construction yields valid quantum measurements~\cite{Belavkin1975,HausladenWootters1994,EldarForney2002}. Defining the POVM elements $M_\mathcal{X}^c := \outp{v_\mathcal{X}^c}{v_\mathcal{X}^c}$, each $M_\mathcal{X}^c$ is positive semidefinite, and the elements sum to the projector onto the span of the seed vectors:
\begin{align}\label{eq:povm-completeness}
\sum_{c \in \mathcal{X}} M_\mathcal{X}^c = \Pi_\mathcal{X},
\end{align}
where $\Pi_\mathcal{X}$ is the orthogonal projector onto $\mathrm{span}\{\ket{\phi_c}\}_{c \in \mathcal{X}}$. When the seed vectors span $\Complex^d$ (in particular, when they are linearly independent), we have $\Pi_\mathcal{X} = I$ and the POVM elements sum to the identity. In the general case where $\Pi_\mathcal{X} \neq I$, we complete the POVM by adding an element $M_\mathcal{X}^\perp := I - \Pi_\mathcal{X}$ corresponding to an ``inconclusive'' outcome; obtaining this outcome would indicate a failure to coordinate on any channel in the safe set. This construction is precisely the \emph{Pretty Good Measurement}~\cite{Belavkin1975,HausladenWootters1994,Barnum2002}, which reduces to the projective measurement in the L\"owdin basis when the Gram matrix is positive definite.

When the Gram matrix $G_\mathcal{X}$ is positive definite, the pseudoinverse reduces to the ordinary inverse, and the measurement vectors simplify to
\begin{align}\label{eq:lowdin-def-pd}
\ket{v_\mathcal{X}^c} = \sum_{c' \in \mathcal{X}} [G_\mathcal{X}^{-1/2}]_{c',c} \ket{\phi_{c'}}, \quad c \in \mathcal{X}.
\end{align}
In this case, the vectors $\{\ket{v_\mathcal{X}^c}\}_{c \in \mathcal{X}}$ are orthonormal:
\begin{align}\label{eq:lowdin-orthonormality}
\inp{v_\mathcal{X}^c}{v_\mathcal{X}^{c'}} = [G_\mathcal{X}^{-1/2} G_\mathcal{X} G_\mathcal{X}^{-1/2}]_{c,c'} = \delta_{c,c'},
\end{align}
and the measurement is projective. This is the classical \emph{L\"owdin symmetric orthonormalization}~\cite{Lowdin1950}, which is distinguished from alternatives such as Gram--Schmidt by a fundamental optimality property: among all orthonormal bases $\{\ket{u_c}\}_{c \in \mathcal{X}}$ for the subspace spanned by $\{\ket{\phi_c}\}_{c \in \mathcal{X}}$, the L\"owdin vectors uniquely minimize the total squared distance $\sum_{c} \| \ket{u_c} - \ket{\phi_c} \|^2$ from the original seed vectors~\cite{Lowdin1950,Mayer2003}. This optimality, together with the symmetric treatment of all seed vectors, makes L\"owdin orthonormalization natural for our construction.

\begin{remark}[Frame operator representation]\label{rem:frame-operator}
When the seed vectors $\{\ket{\phi_c}\}_{c \in \mathcal{X}}$ are linearly independent, the Gram matrix is related to the \emph{frame operator} $S_\mathcal{X} := \sum_{c \in \mathcal{X}} \proj{\phi_c}$ by $G_\mathcal{X} = \Phi_\mathcal{X}^\dagger \Phi_\mathcal{X}$ and $S_\mathcal{X} = \Phi_\mathcal{X} \Phi_\mathcal{X}^\dagger$, where $\Phi_\mathcal{X}$ is the $d \times |\mathcal{X}|$ matrix whose columns are the seed vectors $\{\ket{\phi_c}\}_{c \in \mathcal{X}}$. The L\"owdin vectors then admit the equivalent representation
\begin{align}\label{eq:lowdin-frame-op}
\ket{v_\mathcal{X}^c} = S_\mathcal{X}^{-1/2}\ket{\phi_c}, \qquad c \in \mathcal{X},
\end{align}
which we use in later sections when analyzing specific frame constructions (see in particular Sections~\ref{sec:random-frames} and~\ref{sec:harmonic}).
\end{remark}

For later analysis of explicit types of seed vectors, the following proposition can be convenient, as it simplifies calculations for seed vectors that enjoy symmetries:
\begin{proposition}[Symmetry preservation]\label{prop:lowdin-symmetry}
If the seed vectors are invariant under a unitary $U$ that permutes the indices, i.e., $U\ket{\phi_c} = \ket{\phi_{\sigma(c)}}$ for some permutation $\sigma$, then
\begin{align}\label{eq:lowdin-equivariance}
U\ket{v_\mathcal{X}^c} = \ket{v_{\sigma(\mathcal{X})}^{\sigma(c)}}.
\end{align}
\end{proposition}
\begin{proof}
The hypothesis implies $[G_{\sigma(\mathcal{X})}]_{\sigma(c),\sigma(c')} = [G_\mathcal{X}]_{c,c'}$, so the Gram matrices are related by simultaneous permutation of rows and columns. Since matrix functions commute with this relabeling, we have $[(G_{\sigma(\mathcal{X})}^+)^{1/2}]_{\sigma(c),\sigma(c')} = [(G_\mathcal{X}^+)^{1/2}]_{c,c'}$. Applying $U$ to Eq.~\eqref{eq:lowdin-def} then yields Eq.~\eqref{eq:lowdin-equivariance}.
\end{proof}

\subsubsection{Conditions for positive definiteness}

While the pseudoinverse construction applies to arbitrary seed vectors, it is useful to have conditions ensuring the Gram matrix is positive definite, so that the measurement is projective.

A central parameter controlling the analysis is the \emph{coherence} of the frame~\cite{welch1974,FoucartRauhut2013}:
\begin{align}\label{eq:coherence-def}
\delta := \max_{a \neq b} |\inp{\phi_a}{\phi_b}|.
\end{align}
The coherence measures the maximum correlation between distinct seed vectors.

The following lemma ensures that the Gram matrix is invertible for small coherence. This result follows from the classical Gershgorin circle theorem~\cite{Gershgorin1931,HornJohnson2013}.

\begin{lemma}[Positive definiteness of the Gram matrix]\label{lem:gram-pd}
If the coherence satisfies $\delta < 1/(d-1)$, then for every safe set $\mathcal{X}$ with $|\mathcal{X}| = d$, the Gram matrix $G_\mathcal{X}$ is positive definite with eigenvalues in the interval $[1 - (d-1)\delta, 1 + (d-1)\delta]$.
\end{lemma}

\begin{proof}
Since each seed vector is normalized, the diagonal entries of $G_\mathcal{X}$ satisfy $[G_\mathcal{X}]_{c,c} = 1$. The off-diagonal entries satisfy $|[G_\mathcal{X}]_{c,c'}| \leq \delta$ for $c \neq c'$. By Gershgorin's theorem~\cite{Gershgorin1931,HornJohnson2013}, every eigenvalue $\lambda$ of $G_\mathcal{X}$ lies in the union of disks
\begin{align}
\bigcup_{c \in \mathcal{X}} \left\{ z \in \mathbb{C} : |z - 1| \leq \sum_{c' \in \mathcal{X} \setminus \{c\}} |[G_\mathcal{X}]_{c,c'}| \right\}.
\end{align}
Since $G_\mathcal{X}$ is Hermitian, its eigenvalues are real. Each disk has center $1$ and radius at most $(d-1)\delta$, so $\lambda \in [1 - (d-1)\delta, 1 + (d-1)\delta]$. When $\delta < 1/(d-1)$, the lower bound is positive, establishing positive definiteness.
\end{proof}

\subsection{Decomposing the quantum winning probability}\label{sec:quantum-structure}

We now analyze the structure of the quantum winning probability, expressing it in terms of overlaps between L\"owdin vectors organized by the intersection size of their associated safe sets. This decomposition separates the combinatorial structure of the game from the geometric properties of the seed frame.

\subsubsection{Decomposition by intersection size}

Our decomposition result reduces the problem of analyzing $\omega_q$ to understanding how L\"owdin vectors for the same channel differ across safe sets with varying overlap.

\begin{theorem}[Decomposition]\label{thm:general-bound}
Let $\{\ket{\phi_c}\}_{c \in [n]}$ be a frame of $n$ unit vectors in $\Complex^d$ such that every $d$-subset is linearly independent. Suppose that for each intersection size $j \in \{1, \ldots, d\}$, there exists a lower bound $L_j \geq 0$ such that
\begin{align}\label{eq:Lj-condition}
|\inp{v_\mathcal{X}^c}{v_\mathcal{Y}^c}|^2 \geq L_j \quad \text{whenever } |\mathcal{X} \cap \mathcal{Y}| = j \text{ and } c \in \mathcal{X} \cap \mathcal{Y}.
\end{align}
Then the quantum strategy using L\"owdin-orthonormalized measurements achieves winning probability
\begin{align}\label{eq:omega-q-general-bound}
\omega_q \geq \frac{d}{n} \sum_{j=1}^{d} p(j) \cdot L_j,
\end{align}
where $p(j)$ is the probability that two uniformly random safe sets containing a fixed channel have intersection size exactly $j$, given by the hypergeometric distribution~\cite{JohnsonKotz1969}
\begin{align}\label{eq:hypergeometric-def}
p(j) = \frac{\binom{d-1}{j-1}\binom{k}{d-j}}{\binom{n-1}{d-1}}.
\end{align}
\end{theorem}

When $2d > n$, the minimum intersection size is $j_{\min} = 2d - n > 1$, and $p(j) = 0$ for $j < j_{\min}$ (since $\binom{k}{d-j} = 0$ when $d - j > k$); the sum effectively begins at $j_{\min}$. When the bounds $L_j$ are \emph{exact} (i.e., $|\inp{v_\mathcal{X}^c}{v_\mathcal{Y}^c}|^2 = L_j$ for all pairs with $|\mathcal{X} \cap \mathcal{Y}| = j$), the inequality becomes an equality.

This result \emph{separates} the combinatorics of the jamming game (encoded in $p(j)$) from the frame-theoretic analysis (encoded in the bounds $L_j$). The proof proceeds in three steps: (i) partition $\omega_q$ as a sum over pairs of safe sets grouped by intersection size; (ii) count the pairs with each intersection size; (iii) apply the bounds $L_j$ and simplify.

\subsubsection{Organizing by intersection size}

From Eq.~\eqref{eq:inTermsOfA}, the quantum winning probability involves the operators $A_c = \sum_{\mathcal{X} \ni c} \outp{v_\mathcal{X}^c}{v_\mathcal{X}^c}$. Our first lemma rewrites this in a form that exposes the intersection structure.

\begin{lemma}[Partition by intersection size]\label{lem:partition-structure}
The quantum winning probability can be written as
\begin{align}\label{eq:omega-partition}
\omega_q = \frac{1}{d|\mathcal{S}|^2} \sum_{c \in [n]} \sum_{j=1}^{d} \sum_{\substack{\mathcal{X}, \mathcal{Y} \ni c \\ |\mathcal{X} \cap \mathcal{Y}| = j}} |\inp{v_\mathcal{X}^c}{v_\mathcal{Y}^c}|^2,
\end{align}
where the innermost sum runs over ordered pairs of safe sets $(\mathcal{X}, \mathcal{Y})$ both containing channel $c$ with intersection size exactly $j$.
\end{lemma}

\begin{proof}
Starting from Eq.~\eqref{eq:inTermsOfA}, we have
\begin{align}\label{eq:omega-from-Ac}
\omega_q = \frac{1}{d|\mathcal{S}|^2} \sum_{c \in [n]} \tr[A_c^2].
\end{align}
Expanding the definition $A_c = \sum_{\mathcal{X} \ni c} \outp{v_\mathcal{X}^c}{v_\mathcal{X}^c}$ and using linearity of the trace together with the identity $\tr[\outp{u}{u}\outp{w}{w}] = |\inp{u}{w}|^2$:
\begin{align}\label{eq:trace-as-overlaps}
\tr[A_c^2] = \sum_{\mathcal{X}, \mathcal{Y} \ni c} |\inp{v_\mathcal{X}^c}{v_\mathcal{Y}^c}|^2.
\end{align}
Since $c$ lies in both $\mathcal{X}$ and $\mathcal{Y}$, the intersection satisfies $1 \leq |\mathcal{X} \cap \mathcal{Y}| \leq d$. Partitioning the sum according to intersection size $j$ and substituting into Eq.~\eqref{eq:omega-from-Ac} yields Eq.~\eqref{eq:omega-partition}.
\end{proof}

\subsubsection{Counting pairs by intersection size}

To evaluate the sum in Eq.~\eqref{eq:omega-partition}, we need the number of pairs with each intersection size.

\begin{lemma}[Intersection size distribution]\label{lem:intersection-counts}
For the $(n,k)$-jamming game with $d = n-k$, let $N_c(j)$ denote the number of ordered pairs $(\mathcal{X}, \mathcal{Y})$ of safe sets containing channel $c$ with $|\mathcal{X} \cap \mathcal{Y}| = j$. Then:
\begin{align}\label{eq:Ncj-formula}
N_c(j) = \binom{d-1}{j-1}\binom{k}{d-j}\binom{n-1}{d-1}.
\end{align}
This count is independent of $c$ by symmetry. The probability $p(j) := N_c(j)/\binom{n-1}{d-1}^2$ is the hypergeometric distribution of Eq.~\eqref{eq:hypergeometric-def}.

The expected size of the symmetric difference $\mathcal{X} \triangle \mathcal{Y} = (\mathcal{X} \setminus \mathcal{Y}) \cup (\mathcal{Y} \setminus \mathcal{X})$ is
\begin{align}\label{eq:expected-sym-diff}
\mathbb{E}[|\mathcal{X} \triangle \mathcal{Y}|] = \frac{2(d-1)k}{n-1}.
\end{align}
For fixed $d$ and large $n$, this approaches $2(d-1)$, meaning typical safe set pairs differ in nearly all non-fixed positions.
\end{lemma}

\begin{proof}
A safe set containing $c$ is determined by choosing $d-1$ additional elements from the remaining $n-1$ channels, giving $\binom{n-1}{d-1}$ choices for $\mathcal{X}$.

For a fixed $\mathcal{X} \ni c$, we count safe sets $\mathcal{Y} \ni c$ with $|\mathcal{X} \cap \mathcal{Y}| = j$. Besides $c$, the intersection must contain $j-1$ elements chosen from $\mathcal{X} \setminus \{c\}$ (size $d-1$), and the remaining $d-j$ elements of $\mathcal{Y}$ must come from $[n] \setminus \mathcal{X}$ (size $k$). This gives $\binom{d-1}{j-1}\binom{k}{d-j}$ choices for $\mathcal{Y}$. Multiplying by the choices for $\mathcal{X}$ yields Eq.~\eqref{eq:Ncj-formula}.

For the expected symmetric difference, note $|\mathcal{X} \triangle \mathcal{Y}| = 2(d-J)$ where $J$ is a random variable. By standard hypergeometric moment formulas~\cite{JohnsonKotz1969}, $\mathbb{E}[J] = 1 + (d-1)^2/(n-1)$, so
\begin{align}
\mathbb{E}[|\mathcal{X} \triangle \mathcal{Y}|] = 2\mathbb{E}[d-J] = 2\left(d - 1 - \frac{(d-1)^2}{n-1}\right) = \frac{2(d-1)k}{n-1}.
\end{align}
\end{proof}

\subsubsection{Proof of Theorem~\ref{thm:general-bound}}

\begin{proof}[Proof of Theorem~\ref{thm:general-bound}]
Applying the assumed bounds $|\inp{v_\mathcal{X}^c}{v_\mathcal{Y}^c}|^2 \geq L_j$ to Eq.~\eqref{eq:omega-partition}:
\begin{align}
\omega_q &\geq \frac{1}{d|\mathcal{S}|^2} \sum_{c \in [n]} \sum_{j=1}^{d} N_c(j) \cdot L_j. \label{eq:omega-with-Lj}
\end{align}
By Lemma~\ref{lem:intersection-counts}, $N_c(j)$ is independent of $c$, so the sum over $c$ contributes a factor of $n$:
\begin{align}
\omega_q &\geq \frac{n}{d|\mathcal{S}|^2} \sum_{j=1}^{d} N_c(j) \cdot L_j. \label{eq:omega-factor-n}
\end{align}
Using $|\mathcal{S}| = \binom{n}{d} = \frac{n}{d}\binom{n-1}{d-1}$ and the definition $p(j) = N_c(j)/\binom{n-1}{d-1}^2$:
\begin{align}
\frac{n}{d|\mathcal{S}|^2} \cdot N_c(j) = \frac{n}{d \cdot \frac{n^2}{d^2}\binom{n-1}{d-1}^2} \cdot N_c(j) = \frac{d}{n} \cdot p(j).
\end{align}
Substituting into Eq.~\eqref{eq:omega-factor-n} yields Eq.~\eqref{eq:omega-q-general-bound}.
\end{proof}

%% file: randomFrames.tex
\section{Random Seed Frames}\label{sec:random-frames}

We now choose seed vectors randomly from the Haar measure. Our goal is to show that a quantum advantage always exists for any number of safe channels $d$, as long as $n$ is sufficiently large.  Specifically, we will choose a \emph{random seed frame} consisting of $n$ unit vectors $\{|\phi_c\rangle\}_{c=0}^{n-1}$ in $\Complex^d$, drawn independently from the Haar measure on the unit sphere~\cite{Meckes2019}. Note that by the unitary invariance of the Haar measure, we may fix one seed vector without loss of generality, which we choose to be the first vector $|\phi_0\rangle = |0\rangle$. The remaining seed vectors $|\phi_1\rangle, \ldots, |\phi_{n-1}\rangle$ are independent Haar-random.
For a safe set $X$ containing channel $0$, the frame operator $S_X$ and the L\"owdin vectors $\{|v^c_X\rangle\}_{c \in X}$ are then as defined in Section~\ref{sec:lowdin}. We write $|v^0_X\rangle = S_X^{-1/2}|0\rangle$ for the L\"owdin vector corresponding to channel $0$.
\begin{remark}[Linear independence almost surely]\label{rem:random-linear-independence}
For $d$ Haar-random unit vectors in $\Complex^d$, the probability of linear dependence is zero. Thus $S_X$ is almost surely positive definite, $S_X^{-1/2}$ is almost surely well-defined, and the measurement is projective almost surely. The L\"owdin vectors are unit vectors by Eq.~\eqref{eq:lowdin-orthonormality}.
\end{remark}

\subsection{Quantum Winning Probability}\label{sec:random-winning}

We first show following the ideas of Theorem~\ref{thm:general-bound} that the expected quantum winning probability can be decomposed into a sum of expectations that depend on the intersection size $|X \cap Y| = j$ between safe sets.
Proposition~\ref{prop:expected-structural} reduces the problem to evaluating the expected overlaps $\mathcal{L}_j$. Two observations simplify the analysis considerably. First, the expected overlaps $\mathcal{L}_j$ depend only on $d$, not on $n$, since they are determined by the Haar measure on $\Complex^d$. Second, as $n \to \infty$ with $d$ fixed, the hypergeometric weights concentrate on $j = 1$: we have $p(1) \to 1$ and $p(j) \to 0$ for $j > 1$, so that $\mathcal{L} \to \mathcal{L}_1$. The large-$n$ behavior of $\E[\omega_q]$ is thus controlled entirely by $\mathcal{L}_1$, the expected overlap for minimally overlapping safe sets. Our main result evaluates $\mathcal{L}_1$ in closed form and thereby determines $\E[\omega_q]$ to leading order.

\begin{proposition}[Expected decomposition]\label{prop:expected-structural}
For i.i.d.\ Haar-random seed vectors, the expected quantum winning probability is
\begin{align}\label{eq:expected-omega-q}
\E[\omega_q] = \frac{d}{n}\, \mathcal{L}, \qquad \text{where} \quad \mathcal{L} := \sum_{j=j_{\min}}^{d} p(j)\, \mathcal{L}_j, \quad j_{\min} := \max(1,\, 2d - n),
\end{align}
and $\mathcal{L}_j := \E\!\left[|\langle v^0_X | v^0_Y\rangle|^2 \,\big|\, |X \cap Y| = j\right]$ is the expected squared L\"owdin overlap conditioned on intersection size $j$. The expectation in $\mathcal{L}_j$ is taken over both the Haar measure on seed vectors and the uniform distribution over pairs $(X,Y)$ with $|X \cap Y| = j$.
\end{proposition}

\begin{proof}
From Eq.~\eqref{eq:inTermsOfA}, the quantum winning probability for any fixed frame is
$\omega_q = \frac{1}{d|\mathcal{S}|^2} \sum_{c \in [n]} \tr\left[A_c^2\right]$.
For i.i.d.\ Haar-random seed vectors, the joint distribution is invariant under permutation of channel labels. Therefore $\E[\tr[A_c^2]] = \E[\tr[A_0^2]]$ for all $c \in [n]$, and taking expectations:
\begin{align}\label{eq:expected-omega-step1}
\E[\omega_q] = \frac{n}{d|\mathcal{S}|^2} \E[\tr[A_0^2]].
\end{align}
Expanding the trace as in Eq.~\eqref{eq:Aexpand} gives $\tr[A_0^2] = \sum_{X, Y \ni 0} |\langle v^0_X | v^0_Y\rangle|^2$. Using $|\mathcal{S}| = \binom{n}{d} = \frac{n}{d}\binom{n-1}{d-1}$ and partitioning by intersection size $|X \cap Y| = j$ using the distribution from Lemma~\ref{lem:intersection-counts} yields Eq.~\eqref{eq:expected-omega-q}.
\end{proof}

\begin{remark}[Symmetry holds in expectation]\label{rem:symmetry-expectation}
In any fixed realization of the random frame, the values $\tr[A_c^2]$ for different channels $c$ will generally differ. The symmetry argument---that all channels contribute equally---holds only in expectation over the Haar measure, since the joint distribution of seed vectors is permutation-invariant under relabeling of channels. Thus Eq.~\eqref{eq:expected-omega-q} characterizes $\E[\omega_q]$, not $\omega_q$ for a specific realization.
\end{remark}

We are now ready to phrase the winning probability in terms of the alignment parameter $\alpha_d$.

\begin{theorem}[Quantum winning probability for random frames]\label{thm:random-winning}
Let $\{|\phi_c\rangle\}_{c=0}^{n-1}$ be i.i.d.\ Haar-random seed vectors in $\Complex^d$ with $|\phi_0\rangle = |0\rangle$ fixed. Define the \emph{alignment parameter}
\begin{align}\label{eq:alpha-def}
\alpha_d := \E\!\left[\langle 0 | S_X^{-1/2} | 0\rangle^2\right],
\end{align}
where $S_X = \sum_{c \in X} |\phi_c\rangle\langle\phi_c|$ and the expectation is over Haar-random vectors $\{|\phi_a\rangle\}_{a \in X\setminus\{0\}}$. Then for fixed $d \geq 2$, the expected quantum winning probability satisfies
\begin{align}\label{eq:Ewq-asymptotic}
\E[\omega_q] = \frac{d}{n}\left[\alpha_d^2 + \frac{(1-\alpha_d)^2}{d-1}\right] + O(n^{-2}).
\end{align}
The alignment parameter $\alpha_d$ equals the expected success probability of the Pretty Good Measurement~\cite{HausladenWootters1994,Belavkin1975} for $d$ equiprobable Haar-random pure states in $\Complex^d$ (Theorem~\ref{thm:pgm-equivalence}). Through this identification asymptotic bounds on $\alpha_d$ are known:
\begin{align}\label{eq:alpha-values}
\alpha_2 = \frac{5}{6}, \qquad \liminf_{d \to \infty}\, \alpha_d \;\geq\; \frac{64}{9\pi^2} \approx 0.7205,
\end{align}
where the asymptotic lower bound was established by Montanaro~\cite{Montanaro2007} 
	(Corollary~\ref{cor:alpha-asymptotic}).
\end{theorem}

\begin{remark}[Fixing $|\phi_0\rangle = |0\rangle$]\label{rem:wlog-fixing}
The choice $|\phi_0\rangle = |0\rangle$ is without loss of generality: for any realization of i.i.d.\ Haar-random vectors, we may apply a unitary $U$ with $U|\phi_0\rangle = |0\rangle$ to all seed vectors simultaneously. By the invariance of the Haar measure under unitary transformations, the rotated vectors $\{U|\phi_c\rangle\}_{c \neq 0}$ remain i.i.d.\ Haar-random, and the frame operator transforms as $S_X \mapsto U S_X U^\dagger$, leaving $\langle\phi_0|S_X^{-1/2}|\phi_0\rangle = \langle 0|(US_XU^\dagger)^{-1/2}|0\rangle$ unchanged. This argument is used again in the proof of Theorem~\ref{thm:pgm-equivalence}.
\end{remark}

The proof combines two ingredients. The first is the \emph{alignment formula} (Theorem~\ref{thm:master-formula}), which evaluates the  overlap $\mathcal{L}_1$ of minimally overlapping safe sets in closed form as a function of $\alpha_d$. The second is a sandwich bound showing that $\mathcal{L} \to \mathcal{L}_1$ as $n \to \infty$ at rate $O(1/n)$, which follows from the concentration of the hypergeometric weights $p(j)$. We establish the alignment formula and the characterization of $\alpha_d$ in Sections~\ref{sec:random-master-formula}--\ref{sec:random-pgm}, and give the proof of Theorem~\ref{thm:random-winning} in Section~\ref{sec:random-winning-proof}.

\subsubsection{Alignment Formula}\label{sec:random-master-formula}
\begin{theorem}[Alignment formula]\label{thm:master-formula}
For Haar-random seed vectors, the expected squared overlap of L\"owdin vectors from minimally overlapping safe sets ($|X \cap Y| = 1$) is
\begin{align}\label{eq:L1-formula}
\mathcal{L}_1 = \alpha_d^2 + \frac{(1 - \alpha_d)^2}{d-1}.
\end{align}
\end{theorem}

The proof relies on a geometric decomposition of the L\"owdin vector and independence properties of the Haar measure, which we consolidate into three lemmas: we decompose each L\"owdin vector into a component along $|0\rangle$ and an orthogonal component (Lemma~\ref{lem:orthogonal-decomp}); we establish that for disjoint safe sets these components are independent (Lemma~\ref{lem:combined-independence}); and we compute the expected squared overlap using Haar-random inner product statistics (Lemma~\ref{lem:haar-inner-product}).

\begin{lemma}[Orthogonal decomposition]\label{lem:orthogonal-decomp}
Let $X$ be a safe set containing channel~$0$. The L\"owdin vector $|v^0_X\rangle = S_X^{-1/2}|0\rangle$ admits a unique decomposition
\begin{align}\label{eq:lowdin-decomp}
|v^0_X\rangle = A_X |0\rangle + B_X |\hat{v}^\perp_X\rangle,
\end{align}
where $A_X := \langle 0 | S_X^{-1/2} | 0\rangle > 0$ is real, $B_X := \sqrt{1 - A_X^2} \geq 0$, and $|\hat{v}^\perp_X\rangle$ is a unit vector orthogonal to $|0\rangle$.
\end{lemma}

\begin{proof}
Since $S_X^{-1/2}$ is almost surely positive definite (and hence Hermitian with positive eigenvalues), we have $A_X = \langle 0|S_X^{-1/2}|0\rangle > 0$. The L\"owdin vector $|v^0_X\rangle = S_X^{-1/2}|0\rangle$ is a unit vector (Remark~\ref{rem:random-linear-independence}). Write
\begin{align}
|v^0_X\rangle = A_X |0\rangle + |w\rangle,
\end{align}
where $|w\rangle := |v^0_X\rangle - A_X|0\rangle$ is orthogonal to $|0\rangle$ by construction. Setting $B_X := \||w\rangle\|$ and $|\hat{v}^\perp_X\rangle := |w\rangle/B_X$ (when $B_X > 0$) completes the decomposition. The identity $A_X^2 + B_X^2 = 1$ follows from $\||v^0_X\rangle\| = 1$.
\end{proof}

\begin{lemma}[Independence and Haar distribution for disjoint safe sets]\label{lem:combined-independence}
Let $X$ and $Y$ be safe sets containing channel $0$ with $X \cap Y = \{0\}$. Then:
\begin{enumerate}[label=(\roman*)]
\item \textbf{Independence:} The tuples $(A_X, B_X, |\hat{v}^\perp_X\rangle)$ and $(A_Y, B_Y, |\hat{v}^\perp_Y\rangle)$ are independent.
\item \textbf{Haar distribution:} The vectors $|\hat{v}^\perp_X\rangle$ and $|\hat{v}^\perp_Y\rangle$ are independent Haar-random unit vectors in $|0\rangle^\perp \cong \Complex^{d-1}$.
\item \textbf{Coefficient independence:} Each of $(A_X, B_X)$ and $(A_Y, B_Y)$ is independent of its corresponding perpendicular direction $|\hat{v}^\perp_X\rangle$ and $|\hat{v}^\perp_Y\rangle$, respectively.
\end{enumerate}
\end{lemma}

\begin{proof}
\textbf{Part (i):} Since $X \cap Y = \{0\}$, the collections $\{|\phi_a\rangle\}_{a \in X\setminus\{0\}}$ and $\{|\phi_b\rangle\}_{b \in Y\setminus\{0\}}$ are disjoint and consist of independent Haar-random vectors. The tuple $(A_X, B_X, |\hat{v}^\perp_X\rangle)$ is a function of $\{|\phi_a\rangle\}_{a \in X\setminus\{0\}}$ alone, and $(A_Y, B_Y, |\hat{v}^\perp_Y\rangle)$ is a function of $\{|\phi_b\rangle\}_{b \in Y\setminus\{0\}}$ alone. Independence follows.

\textbf{Part (ii):} Let $\mathcal{U} := \{U \in U(d) : U|0\rangle = |0\rangle\}$ be the stabilizer of $|0\rangle$, which acts as $U(d-1)$ on $|0\rangle^\perp$. For any $U \in \mathcal{U}$, applying $U$ to all seed vectors $|\phi_a\rangle \mapsto U|\phi_a\rangle$ preserves their joint distribution (by Haar invariance). Under this map, the frame operator transforms as $S_X \mapsto US_XU^\dagger$, and hence $S_X^{-1/2} \mapsto US_X^{-1/2}U^\dagger$ (since the matrix inverse square root commutes with unitary conjugation). The L\"owdin vector transforms as $|v^0_X\rangle = S_X^{-1/2}|0\rangle \mapsto US_X^{-1/2}U^\dagger|0\rangle = US_X^{-1/2}|0\rangle = U|v^0_X\rangle$, where we used $U^\dagger|0\rangle = |0\rangle$. It follows that the parallel coefficient $A_X = \langle 0|v^0_X\rangle$ is invariant (since $\langle 0|U|v^0_X\rangle = \langle 0|v^0_X\rangle$), and consequently $B_X = \sqrt{1 - A_X^2}$ is also invariant, while the perpendicular component transforms as $|\hat{v}^\perp_X\rangle \mapsto U|\hat{v}^\perp_X\rangle$.

The distribution of $|\hat{v}^\perp_X\rangle$ is therefore invariant under all unitaries in $\mathcal{U} \cong U(d-1)$. The unique probability measure on the unit sphere in $\Complex^{d-1}$ with this invariance is the Haar measure. Thus $|\hat{v}^\perp_X\rangle$ is Haar-random in $|0\rangle^\perp$. The same argument applies to $|\hat{v}^\perp_Y\rangle$, and independence follows from Part (i).

\textbf{Part (iii):} We show that the conditional distribution of $|\hat{v}^\perp_X\rangle$ given $(A_X, B_X) = (a, b)$ is Haar measure on $|0\rangle^\perp$, independently of the values $a, b$. Fix any $U \in \mathcal{U}$. As shown in Part~(ii), the map $|\phi_c\rangle \mapsto U|\phi_c\rangle$ preserves the joint distribution of all seed vectors, leaves $(A_X, B_X)$ invariant, and sends $|\hat{v}^\perp_X\rangle \mapsto U|\hat{v}^\perp_X\rangle$. Since $(A_X, B_X)$ is unchanged by the map, conditioning on $(A_X, B_X) = (a, b)$ is compatible with it: the conditional distribution of $|\hat{v}^\perp_X\rangle$ given $(A_X, B_X) = (a, b)$ is invariant under $U$. Since this holds for every $U \in \mathcal{U} \cong U(d-1)$, the conditional distribution is Haar measure on the unit sphere in $|0\rangle^\perp$, which does not depend on $(a, b)$. By definition, $|\hat{v}^\perp_X\rangle$ and $(A_X, B_X)$ are therefore independent.
\end{proof}

\begin{lemma}[Inner product of Haar-random vectors]\label{lem:haar-inner-product}
For independent Haar-random unit vectors $|u\rangle, |w\rangle$ in $\Complex^{m}$, writing $\langle u|w\rangle = Re^{i\psi}$ with $R \geq 0$:
\begin{enumerate}[label=(\roman*)]
\item $\E[R^2] = 1/m$.
\item The phase $\psi$ is uniform on $[0, 2\pi)$ and independent of $R$.
\end{enumerate}
\end{lemma}

\begin{proof}
By unitary invariance, we may fix $|u\rangle = |0\rangle$. Then $\langle u|w\rangle = w_1$, the first component of $|w\rangle$. For a Haar-random unit vector in $\Complex^m$, the squared magnitude $|w_1|^2$ follows a $\mathrm{Beta}(1, m-1)$ distribution with mean $1/m$~\cite{Meckes2019}. Since the distribution of $|w\rangle$ is invariant under $|w\rangle \mapsto e^{i\theta}|w\rangle$, the phase of $w_1$ is uniform on $[0, 2\pi)$ and independent of $|w_1|$.
\end{proof}

\begin{proof}[Proof of Theorem~\ref{thm:master-formula}]
Using the decomposition from Lemma~\ref{lem:orthogonal-decomp}:
\begin{align}
\langle v^0_X | v^0_Y\rangle = A_X A_Y + B_X B_Y \langle\hat{v}^\perp_X | \hat{v}^\perp_Y\rangle.
\end{align}
Writing $\langle\hat{v}^\perp_X | \hat{v}^\perp_Y\rangle = Re^{i\psi}$:
\begin{align}\label{eq:overlap-squared}
|\langle v^0_X | v^0_Y\rangle|^2 &= |A_X A_Y + B_X B_Y Re^{i\psi}|^2 \nonumber \\
&= A_X^2 A_Y^2 + B_X^2 B_Y^2 R^2 + 2A_X A_Y B_X B_Y R \cos\psi.
\end{align}

By Lemma~\ref{lem:combined-independence}, the pairs $(A_X, B_X)$ and $(A_Y, B_Y)$ are mutually independent and independent of $\langle\hat{v}^\perp_X | \hat{v}^\perp_Y\rangle$. By Lemma~\ref{lem:haar-inner-product} with $m = d-1$, $R$ and $\psi$ are independent with
\begin{align}
\E[\cos\psi] = 0, \qquad \E[R^2] = \frac{1}{d-1}.
\end{align}

Taking expectations of Eq.~\eqref{eq:overlap-squared}:
\begin{align}
\mathcal{L}_1 &= \E[A_X^2]\E[A_Y^2] + \E[B_X^2]\E[B_Y^2] \cdot \frac{1}{d-1} + 0 \nonumber \\
&= \alpha_d^2 + (1 - \alpha_d)^2 \cdot \frac{1}{d-1},
\end{align}
using $\E[A^2] = \alpha_d$ and $\E[B^2] = 1 - \E[A^2] = 1 - \alpha_d$ (since $A^2 + B^2 = 1$).
\end{proof}

\subsubsection{Connection to the Pretty Good Measurement}\label{sec:random-pgm}

We establish a connection between $\alpha_d$ and the Pretty Good Measurement (PGM) from quantum state discrimination~\cite{HausladenWootters1994,Belavkin1975}.

\begin{theorem}[PGM equivalence]\label{thm:pgm-equivalence}
For Haar-random seed vectors $\{|\phi_a\rangle\}_{a=1}^{d}$ in $\Complex^d$,
\begin{align}\label{eq:pgm-equivalence}
\alpha_d = \E[P_{\mathrm{pgm}}],
\end{align}
where $P_{\mathrm{pgm}} := \frac{1}{d}\sum_{a=1}^{d} |\langle\phi_a|S^{-1/2}|\phi_a\rangle|^2$ is the PGM success probability for the equiprobable ensemble $\{|\phi_a\rangle\}_{a=1}^d$ with $S = \sum_{a=1}^d |\phi_a\rangle\langle\phi_a|$.
\end{theorem}

\begin{proof}
In Eq.~\eqref{eq:alpha-def}, $\alpha_d = \E[\langle 0|S_X^{-1/2}|0\rangle^2]$ where $|\phi_0\rangle = |0\rangle$ is fixed and the remaining $d-1$ vectors are Haar-random. In the PGM, all $d$ vectors $\{|\phi_a\rangle\}_{a=1}^d$ are Haar-random, and the success probability is $P_{\mathrm{pgm}} = \frac{1}{d}\sum_{a=1}^{d} \langle\phi_a|S^{-1/2}|\phi_a\rangle^2$ (noting that $\langle\phi_a|S^{-1/2}|\phi_a\rangle > 0$ since $S^{-1/2}$ is positive definite).

\emph{Step 1: Reduction by symmetry.} Since the $d$ vectors are i.i.d.\ Haar-random, the summands in $P_{\mathrm{pgm}}$ are identically distributed, so $\E[P_{\mathrm{pgm}}] = \E[\langle\phi_1|S^{-1/2}|\phi_1\rangle^2]$.

\emph{Step 2: Conditioning on $|\phi_1\rangle$.} Conditionally on $|\phi_1\rangle = |u\rangle$, the remaining vectors are still i.i.d.\ Haar-random. By the same rotation argument used to fix $|\phi_0\rangle = |0\rangle$ (Remark~\ref{rem:wlog-fixing}), $\E[\langle u|S^{-1/2}|u\rangle^2 \mid |\phi_1\rangle = |u\rangle] = \alpha_d$ for every $|u\rangle$, and taking expectations over $|\phi_1\rangle$ gives $\E[P_{\mathrm{pgm}}] = \alpha_d$.
\end{proof}

\begin{corollary}[Asymptotic lower bound]\label{cor:alpha-asymptotic}
\begin{align}
\liminf_{d \to \infty}\, \alpha_d \;\geq\; \left(\frac{8}{3\pi}\right)^{\!2} = \frac{64}{9\pi^2} \approx 0.7205.
\end{align}
\end{corollary}

\begin{proof}
By Theorem~\ref{thm:pgm-equivalence}, $\alpha_d = \E[P_{\mathrm{pgm}}]$ for $d$ equiprobable Haar-random states in $\Complex^d$. Montanaro~\cite{Montanaro2007} proved this lower bound using random matrix theory: the empirical eigenvalue distribution of the frame operator $S = \sum_{c} |\phi_c\rangle\langle\phi_c|$ converges to the Mar\v{c}enko--Pastur law~\cite{MarcenkoPastur1967,Meckes2019} with parameter $\gamma = 1$, and a trace-norm comparison with the resulting asymptotic PGM yields $\liminf_{d\to\infty} \E[P_{\mathrm{pgm}}] \geq [8/(3\pi)]^2$.
\end{proof}

\subsubsection{Proof of Theorem~\ref{thm:random-winning}}\label{sec:random-winning-proof}

\begin{proof}[Proof of Theorem~\ref{thm:random-winning}]
By Proposition~\ref{prop:expected-structural}, $\E[\omega_q] = (d/n)\,\mathcal{L}$ with $\mathcal{L} = \sum_{j=j_{\min}}^{d} p(j)\,\mathcal{L}_j$. We show that $\mathcal{L} = \mathcal{L}_1 + O(1/n)$, which combined with the alignment formula (Theorem~\ref{thm:master-formula}) yields Eq.~\eqref{eq:Ewq-asymptotic}.

Since each $\mathcal{L}_j$ is an expected squared overlap of unit vectors, we have $0 \leq \mathcal{L}_j \leq 1$ for all $j$. This gives the sandwich bound
\begin{align}\label{eq:L-sandwich-sec42}
p(1)\,\mathcal{L}_1 \;\leq\; \mathcal{L} \;\leq\; p(1)\,\mathcal{L}_1 + (1 - p(1)),
\end{align}
where the lower bound drops the non-negative terms $p(j)\mathcal{L}_j$ for $j \geq 2$, and the upper bound replaces $\mathcal{L}_j \leq 1$. Both bounds differ from $\mathcal{L}_1$ by $O(1 - p(1))$. A direct calculation from the hypergeometric distribution (Lemma~\ref{lem:intersection-counts}) gives
\begin{align}
p(1) = \prod_{i=0}^{d-2} \frac{n - d - i}{n - 1 - i} = 1 - \frac{(d-1)^2}{n} + O(n^{-2}).
\end{align}
Therefore $\mathcal{L} = \mathcal{L}_1 + O(1/n)$ for fixed $d$, and substituting the alignment formula Eq.~\eqref{eq:L1-formula} into $\E[\omega_q] = (d/n)\mathcal{L}$ yields Eq.~\eqref{eq:Ewq-asymptotic}. The characterization of $\alpha_d$ via the PGM is established in Theorem~\ref{thm:pgm-equivalence}, and the asymptotic limit in Corollary~\ref{cor:alpha-asymptotic}.
\end{proof}

\subsection{Quantum Advantage}\label{sec:random-advantage}

The main result of this section establishes that random seed vectors achieve quantum advantage for all sufficiently large dimensions, provided the number of channels is sufficiently large.

\begin{theorem}[Quantum advantage]\label{thm:quantum-advantage-general}
For all sufficiently large $d$, there exists $N(d)$ such that for all $n \geq N(d)$, random seed vectors achieve expected quantum advantage:
\begin{align}\label{eq:quantum-advantage-general}
\E[\omega_q] > \omega_c(n,\, n-d).
\end{align}
Moreover, the asymptotic advantage ratio satisfies
\begin{align}\label{eq:asymptotic-ratio}
\lim_{n \to \infty} \frac{\E[\omega_q]}{\omega_c(n,\, n-d)} = \frac{(2d-1)\,\mathcal{L}_1(d)}{d},
\end{align}
which exceeds~$1$ for all sufficiently large~$d$.
\end{theorem}

\subsubsection{\texorpdfstring{Convergence of $\mathcal{L}$ to $\mathcal{L}_1$}{Convergence of L to L1}}\label{sec:L-convergence}

The convergence $\mathcal{L} \to \mathcal{L}_1$ as $n \to \infty$ was established in the proof of Theorem~\ref{thm:random-winning} via the sandwich bound (Eq.~\eqref{eq:L-sandwich-sec42}). For later use, we record the lower bound separately.

\begin{proposition}[Lower bound on $\mathcal{L}$]\label{prop:L-lower-bound}
For $n \geq 2d - 1$ (so that $j_{\min} = 1$), the expected overlap satisfies
\begin{align}\label{eq:L-lower-bound}
\mathcal{L} \geq p(1) \cdot \mathcal{L}_1,
\end{align}
where $p(1) = \binom{k}{d-1}/\binom{n-1}{d-1}$ is the probability that two random safe sets containing a fixed channel intersect in exactly that channel.
\end{proposition}

\begin{proof}
Since $\mathcal{L}_j \geq 0$ for all $j$ (each $\mathcal{L}_j$ is an expected squared magnitude), we have
$\mathcal{L} = \sum_{j=1}^{d} p(j)\,\mathcal{L}_j \geq p(1)\,\mathcal{L}_1$.
\end{proof}

\begin{remark}[Monotonicity conjecture]\label{rem:monotonicity-conjecture}
Numerical experiments strongly suggest the monotonicity $\mathcal{L}_j < \mathcal{L}_{j+1}$ for all $j$, which would strengthen Eq.~\eqref{eq:L-lower-bound} to the tighter bound $\mathcal{L} \geq \mathcal{L}_1$.  The proof of Theorem~\ref{thm:quantum-advantage-general} does not require monotonicity.
\end{remark}

\subsubsection{Proof of Theorem~\ref{thm:quantum-advantage-general}}\label{sec:advantage-proof}

\begin{proof}[Proof of Theorem~\ref{thm:quantum-advantage-general}]
By Proposition~\ref{prop:expected-structural}, $\E[\omega_q] = (d/n)\,\mathcal{L}$. By the sandwich bound established in the proof of Theorem~\ref{thm:random-winning} (Eq.~\eqref{eq:L-sandwich-sec42}), $\mathcal{L} \to \mathcal{L}_1$ as $n \to \infty$ with $d$ fixed. The classical value satisfies $\omega_c(n, n-d) = d^2/((2d-1)n) + O(n^{-2})$ (Proposition~\ref{prop:classical-scaling}), so
\begin{align}
\lim_{n\to\infty}\frac{\E[\omega_q]}{\omega_c(n,\, n-d)} = \frac{(2d-1)\,\mathcal{L}_1}{d}.
\end{align}

By the alignment formula (Theorem~\ref{thm:master-formula}), $\mathcal{L}_1 = \alpha_d^2 + (1 - \alpha_d)^2/(d-1)$. As $d \to \infty$, the second term vanishes and $d/(2d-1) \to 1/2$, so the ratio approaches $2\alpha_d^2$. By Corollary~\ref{cor:alpha-asymptotic}, $\liminf_{d\to\infty} \alpha_d \geq 64/(9\pi^2) \approx 0.7205$, giving
\begin{align}
\liminf_{d\to\infty} \frac{(2d-1)\,\mathcal{L}_1}{d} \;\geq\; 2\!\left(\frac{64}{9\pi^2}\right)^{\!2} \approx 1.038 > 1.
\end{align}
Therefore $(2d-1)\mathcal{L}_1/d > 1$ for all sufficiently large $d$. Since $\E[\omega_q]/\omega_c$ converges to this limit, there exists $N(d)$ such that $\E[\omega_q] > \omega_c$ for all $n \geq N(d)$.
\end{proof}

\begin{remark}[Equality of the fixed-$d$ limit]
For each fixed safe-set size $d \geq 2$, the asymptotic advantage ratio
\begin{align}
r(d) := \lim_{n\to\infty}\frac{\E[\omega_q]}{\omega_c(n, n-d)} = \frac{(2d-1)\mathcal{L}_1(d)}{d}
\end{align}
is an exact equality, not merely a lower bound. This follows from the sandwich bound established in the proof of Theorem~\ref{thm:random-winning} (Eq.~\eqref{eq:L-sandwich-sec42}), which gives $\E[\omega_q] = (d/n)\mathcal{L}_1(d) + O(n^{-2})$ at fixed $d$, combined with the classical scaling $\omega_c(n,n-d) = d^2/((2d-1)n) + O(n^{-2})$ from Proposition~\ref{prop:classical-scaling}. Taking the ratio and passing to the limit in $n$ yields Eq.~\eqref{eq:asymptotic-ratio} as an equality.
\end{remark}

\begin{remark}[Order of limits in the large-$d$ regime]
The quantum advantage claim of Theorem~\ref{thm:quantum-advantage-general} involves an \emph{iterated} limit. For each fixed $d$, we first let $n \to \infty$ to obtain $r(d)$ as in Eq.~\eqref{eq:asymptotic-ratio}. We then study the behavior of the sequence $\{r(d)\}_{d \geq 2}$ as $d \to \infty$. Using the alignment formula $\mathcal{L}_1(d) = \alpha_d^2 + (1-\alpha_d)^2/(d-1)$, the second term vanishes and $(2d-1)/d \to 2$, so $r(d) \to 2\alpha_d^2$ up to $O(1/d)$ corrections. Applying Montanaro's lower bound $\liminf_{d\to\infty}\alpha_d \geq (8/(3\pi))^2$ (Corollary~\ref{cor:alpha-asymptotic}) yields
\begin{align}
\liminf_{d\to\infty} r(d) \;\geq\; 2\!\left(\frac{8}{3\pi}\right)^{\!4} \approx 1.038,
\end{align}
which in particular exceeds $1$, establishing the quantum advantage claim.

We emphasize that this is an iterated limit, with $n \to \infty$ taken first at each fixed $d$, and then $d \to \infty$ applied to the resulting sequence $r(d)$. We make no claim about joint limits in which $n$ and $d$ grow simultaneously: the $O(n^{-2})$ error term in Theorem~\ref{thm:random-winning} has a $d$-dependent constant, so uniformity in $d$ would require a separate argument in a regime such as $n \gg d^2$. In addition, direct numerical computation of $r(d)$ verifies $r(d) > 1$ for all $d = 2, \ldots, 50$, indicating that the threshold $d_0$ above which the analytic bound applies is in practice small; we do not, however, make analytic claims about the location of $d_0$.
\end{remark}

For finite $n$, we can give an explicit sufficient condition using the lower bound from Proposition~\ref{prop:L-lower-bound}.

\begin{corollary}[Explicit finite-$n$ bound]\label{cor:finite-n-bound}
For $n \geq 2d - 1$, the expected quantum winning probability satisfies
\begin{align}\label{eq:finite-n-bound}
\E[\omega_q] \geq \frac{d}{n}\,p(1)\,\mathcal{L}_1.
\end{align}
In particular, $\E[\omega_q] > \omega_c(n, n-d)$ whenever
\begin{align}\label{eq:finite-n-threshold}
p(1) > \frac{n\,\omega_c(n, n-d)}{(2d-1)\,\mathcal{L}_1}.
\end{align}
\end{corollary}

\begin{proof}
The bound follows from combining Proposition~\ref{prop:expected-structural} with Proposition~\ref{prop:L-lower-bound}: $\E[\omega_q] = (d/n)\mathcal{L} \geq (d/n)\,p(1)\,\mathcal{L}_1$. The sufficient condition is obtained by requiring this to exceed $\omega_c$.
\end{proof}

\begin{remark}[Asymptotic ratios]\label{rem:asymptotic-ratios}
The asymptotic advantage ratio $(2d-1)\mathcal{L}_1/d$ decreases with $d$ but remains bounded away from $1$:
\begin{center}
\begin{tabular}{c|ccccc}
$d$ & 2 & 3 & 4 & 5 & $\infty$ \\
\hline
$(2d-1)\mathcal{L}_1/d$ & $13/12 \approx 1.083$ & $1.075$ & $1.069$ & $1.065$ & $\geq 1.038$
\end{tabular}
\end{center}
By Corollary~\ref{cor:alpha-asymptotic}, the ratio remains at least $2(64/(9\pi^2))^2 \approx 1.038$ for large $d$, a persistent $\gtrsim\!4\%$ quantum advantage.
\end{remark}

We illustrate the general theorem with exact calculations for small $d$.

\begin{example}[The $(3,1)$-game, $d=2$]\label{ex:d2-advantage}
For $d = 2$ and $n = 3$: $p(1) = p(2) = 1/2$, $\mathcal{L}_1 = 13/18$ (from Eq.~\eqref{eq:alpha-values} and the alignment formula), and $\mathcal{L}_2 = 1$. Thus $\mathcal{L} = 31/36$ and
\begin{align}
\E[\omega_q] = \frac{2}{3} \cdot \frac{31}{36} = \frac{31}{54} \approx 0.574 > \frac{5}{9} \approx 0.556 = \omega_c(3,1).
\end{align}
The advantage ratio is $31/30 \approx 1.033$.
\end{example}

%% file: exampleStrategies.tex
\section{Explicit quantum strategies}\label{sec:exampleStrategies}\label{sec:examples}
We now instantiate our general framework with explicit choices of seed vectors $\{\ket{\Phi_c}\}_c$. 

\input{harmonicStrategy}
\input{simplexStrategy}

\input{otherStrategies}

%% file: harmonicStrategy.tex
\subsection{Harmonic Strategy}\label{sec:harmonic}

We now consider seed vectors that stem from a harmonic frame. Harmonic frames can be explicitly constructed for any $(n, k)$, where we here focus on the case $d=2$, i.e. $(n,k=n-2)$. This regime is of interest since it is accessible to present day quantum hardware technologies creating and storing qubit entanglement, and explicit analytical results can be derived.
A harmonic frame in $\Complex^d$ consists of $n \geq d$ vectors obtained by sampling rows from the $n \times n$ discrete Fourier transform matrix~\cite{Christensen2016,Waldron2018}. Explicitly, for the standard choice $S = \{0, 1, \ldots, d-1\}$, the harmonic frame vectors are
\begin{align}
    \ket{h_c} = \frac{1}{\sqrt{d}} \sum_{j=0}^{d-1} \tau^{cj} \ket{j}, \quad c \in \{0, \ldots, n-1\},
    \label{eq:harmonic-def}
\end{align}
where $\tau = e^{2\pi i/n}$ is the primitive $n$-th root of unity. The inner products between distinct vectors are given by the Dirichlet kernel:
\begin{align}
    \inp{h_{c_1}}{h_{c_2}} = \frac{1}{d} \cdot \frac{\sin(\pi d(c_2 - c_1)/n)}{\sin(\pi(c_2 - c_1)/n)} \cdot e^{i\pi(d-1)(c_2-c_1)/n}.
    \label{eq:harmonic-inner}
\end{align}
For $d = 2$, the harmonic seed vectors in $\Complex^2$ take the simple form
\begin{align}
    \ket{h_c} = \frac{1}{\sqrt{2}} \begin{pmatrix} 1 \\ \tau^c \end{pmatrix},
    \label{eq:harmonic-d2}
\end{align}
and the inner product magnitude between distinct vectors is
\begin{align}
    |\inp{h_a}{h_b}| = \left|\cos\left(\frac{\pi(b-a)}{n}\right)\right|.
    \label{eq:harmonic-d2-inner}
\end{align}
Note that unlike the simplex or SIC-POVM constructions we will look at later, harmonic frames are generally not equiangular: the overlap magnitude $|\inp{h_{c_1}}{h_{c_2}}|$ depends on the difference $c_2 - c_1 \mod n$.

\subsubsection{\texorpdfstring{Exact Evaluation for $d = 2$}{Exact Evaluation for d = 2}}

We now derive an exact analytical expression for the quantum winning probability when using harmonic seed vectors with $d = 2$.

\begin{theorem}[Quantum value for harmonic frames with $d = 2$]\label{thm:harmonic-d2}
For the $(n, n-2)$-jamming game with $d = 2$ using harmonic seed vectors, the quantum winning probability is
\begin{align}
    \omega_q^{\mathrm{harm}}(n) = \frac{1}{n-1} + \frac{2T_n}{n(n-1)^2},
    \label{eq:harmonic-d2-main}
\end{align}
where the weighted cosine sum $T_n$ is defined by
\begin{align}
    T_n := \sum_{\Delta=1}^{n-2} (n - 1 - \Delta) \cos\left(\frac{\pi\Delta}{n}\right).
    \label{eq:T-n-def}
\end{align}
\end{theorem}

Our proof proceeds in three steps: first, we compute the L\"owdin orthonormalized vectors for $d = 2$. second, we derive the cross-context overlap between L\"owdin vectors from different safe sets. finally, we obtain the quantum winning probability by summing over all channel pairs and safe set pairs. Let us start with step 1, captured by the following little lemma.

\begin{lemma}[L\"owdin vectors for $d = 2$]\label{lem:lowdin-d2}
For a safe set $\mathcal{X} = \{a, b\}$ with $a \neq b$, let $r_{ab} := |\inp{h_a}{h_b}| = |\cos(\pi(b-a)/n)|$. The L\"owdin orthonormalized vectors are
\begin{align}
    \ket{v_{\mathcal{X}}^a} &= A_{ab} \ket{h_a} + B_{ab} e^{-i\phi_{ab}} \ket{h_b}, \label{eq:lowdin-va}\\
    \ket{v_{\mathcal{X}}^b} &= B_{ab} e^{i\phi_{ab}} \ket{h_a} + A_{ab} \ket{h_b}, \label{eq:lowdin-vb}
\end{align}
where $\phi_{ab} = \arg\inp{h_a}{h_b}$ is the phase of the inner product, and
\begin{align}
    A_{ab} = \frac{\alpha_{ab} + \beta_{ab}}{2}, \quad B_{ab} = \frac{\alpha_{ab} - \beta_{ab}}{2},
    \label{eq:AB-def}
\end{align}
with $\alpha_{ab} = (1 + r_{ab})^{-1/2}$ and $\beta_{ab} = (1 - r_{ab})^{-1/2}$.
\end{lemma}
\begin{proof}
The Gram matrix for the safe set $\mathcal{X} = \{a, b\}$ is
\begin{align}
    G_{\mathcal{X}} = \begin{pmatrix} 1 & \inp{h_a}{h_b} \\ \inp{h_b}{h_a} & 1 \end{pmatrix} 
    = \begin{pmatrix} 1 & r_{ab} e^{i\phi_{ab}} \\ r_{ab} e^{-i\phi_{ab}} & 1 \end{pmatrix}.
\end{align}
The eigenvalues are $\lambda_\pm = 1 \pm r_{ab}$, with corresponding eigenvectors 
\begin{align}
    \ket{u_+} = \frac{1}{\sqrt{2}}\begin{pmatrix} e^{i\phi_{ab}/2} \\ e^{-i\phi_{ab}/2} \end{pmatrix}, \quad
    \ket{u_-} = \frac{1}{\sqrt{2}}\begin{pmatrix} e^{i\phi_{ab}/2} \\ -e^{-i\phi_{ab}/2} \end{pmatrix}.
\end{align}
The inverse square root is $G_{\mathcal{X}}^{-1/2} = \alpha_{ab} \ket{u_+}\bra{u_+} + \beta_{ab} \ket{u_-}\bra{u_-}$. Computing the matrix elements and applying Definition~\ref{def:lowdin} yields Eqs.~\eqref{eq:lowdin-va}--\eqref{eq:lowdin-vb}.
\end{proof}

Let us now consider step 2 of characterizing the overlaps between measurement vectors for the same channel $c$, but different safe sets $\mathcal{X}$ and $\mathcal{Y}$.
To accomplish this, we first observe that due the cyclic symmetry of the harmonic frame, we can simplify the computation of the quantum winning probability as given in~\eqref{eq:inTermsOfA}.

\begin{lemma}[Cyclic reduction for harmonic frames]\label{lem:cyclic-reduction}
Let $\{|h_c\rangle\}_{c=0}^{n-1}$ be the harmonic frame in $\mathbb{C}^d$ defined in Eq.~\eqref{eq:harmonic-def}, and let $\{|v_{\mathcal{X}}^c\rangle\}$ denote the corresponding L\"owdin vectors. For any channel $c \in [n]$ and any two safe sets $\mathcal{X}, \mathcal{Y}$ with $c \in \mathcal{X} \cap \mathcal{Y}$, define $\mathcal{X}' = \{(x - c) \bmod n : x \in \mathcal{X}\}$ and $\mathcal{Y}' = \{(y - c) \bmod n : y \in \mathcal{Y}\}$. Then
\begin{align}
    |\langle v_{\mathcal{X}}^c | v_{\mathcal{Y}}^c \rangle|^2 = |\langle v_{\mathcal{X}'}^0 | v_{\mathcal{Y}'}^0 \rangle|^2.
    \label{eq:cyclic-reduction}
\end{align}
In particular, without loss of generality we may assume $c = 0$ when computing cross-context overlaps for harmonic frames.
\end{lemma}
\begin{proof}
The harmonic frame vectors satisfy the cyclic symmetry $Z|h_c\rangle = |h_{c+1 \bmod n}\rangle$, where $Z = \sum_{j=0}^{d-1} \tau^j |j\rangle\langle j|$ is the clock operator (a diagonal unitary that multiplies the $j$-th basis vector by $\tau_n^j$). Since $Z$ permutes the channel indices by the cyclic shift $\sigma(c) = c + 1 \bmod n$, Proposition~\ref{prop:lowdin-symmetry} gives
\begin{align}
    Z|v_{\mathcal{X}}^c\rangle = |v_{\mathcal{X}+1}^{c+1}\rangle,
    \label{eq:cyclic-shift-lowdin}
\end{align}
where $\mathcal{X} + 1 = \{x + 1 \bmod n : x \in \mathcal{X}\}$. Applying $Z^{-c}$ (equivalently $Z^{n-c}$) shifts all channel labels by $-c$, mapping channel $c$ to channel $0$ and the safe sets $\mathcal{X}$, $\mathcal{Y}$ to $\mathcal{X}'$, $\mathcal{Y}'$ respectively. Since $Z^{-c}$ is unitary, all inner products are preserved:
\begin{align}
    |\langle v_{\mathcal{X}}^c | v_{\mathcal{Y}}^c \rangle|^2
    = |\langle Z^{-c} v_{\mathcal{X}}^c | Z^{-c} v_{\mathcal{Y}}^c \rangle|^2
    = |\langle v_{\mathcal{X}'}^0 | v_{\mathcal{Y}'}^0 \rangle|^2,
\end{align}
which establishes Eq.~\eqref{eq:cyclic-reduction}.
\end{proof}

\begin{lemma}[Cross-context overlap]\label{lem:cross-overlap-d2}
Let $\mathcal{X} = \{c, j\}$ and $\mathcal{Y} = \{c, k\}$ be two safe sets sharing channel $c$, with $j \neq k$. Define the ``distances from $c$'' as $m_1 = (j - c) \bmod n$ and $m_2 = (k - c) \bmod n$, where $m_1, m_2 \in \{1, \ldots, n-1\}$. Then the cross-context overlap is
\begin{align}
    |\langle v_{\mathcal{X}}^c | v_{\mathcal{Y}}^c \rangle|^2 = \cos^2\!\left(\frac{\pi|m_2 - m_1|}{2n}\right).
    \label{eq:cross-overlap-d2}
\end{align}
\end{lemma}
\begin{proof}
First of all, note that by Lemma~\ref{lem:cyclic-reduction}, we may assume $c = 0$ without loss of generality, so that $\mathcal{X} = \{0, m_1\}$ and $\mathcal{Y} = \{0, m_2\}$.
Using Lemma~\ref{lem:lowdin-d2} with $\mathcal{X} = \{0, m_1\}$ and $\mathcal{Y} = \{0, m_2\}$, the L\"owdin vectors for channel $0$ are
\begin{align}
    \ket{v_{\mathcal{X}}^0} &= A_1 \ket{h_0} + B_1 e^{-i\phi_1} \ket{h_{m_1}}, \\
    \ket{v_{\mathcal{Y}}^0} &= A_2 \ket{h_0} + B_2 e^{-i\phi_2} \ket{h_{m_2}},
\end{align}
where $A_i, B_i$ are defined via Lemma~\ref{lem:lowdin-d2}. Rather than substituting these expressions, we evaluate the overlap by working directly in the standard basis of $\mathbb{C}^2$.
Since $d = 2$, the harmonic frame vectors are $|h_c\rangle = \frac{1}{\sqrt{2}}(1, \, \tau_n^c)^T$. A direct computation gives $[S_{\mathcal{X}}]_{01} = \frac{1}{2}(1 + e^{-2\pi i m_1/n}) = \cos(\pi m_1/n) \cdot e^{-i\pi m_1/n}$. Writing this as $r_1 e^{-i\phi_1}$, the frame operator $S_{\mathcal{X}} = |h_0\rangle\langle h_0| + |h_{m_1}\rangle\langle h_{m_1}|$ takes the form
\begin{align}
    S_{\mathcal{X}} = \begin{pmatrix} 1 & r_1 e^{-i\phi_1} \\ r_1 e^{i\phi_1} & 1 \end{pmatrix},
    \label{eq:frame-op-d2-coord}
\end{align}
where $r_1 := \cos(\pi m_1/n)$ is the \emph{signed} inner product coefficient (note: $r_1$ can be negative for $m_1 > n/2$, unlike the unsigned $r_{ab} = |\inp{h_a}{h_b}|$ of Lemma~\ref{lem:lowdin-d2}) and $\phi_1 := \pi m_1/n$. Since $|h_0\rangle = \frac{1}{\sqrt{2}}(1,1)^T$ has equal components, and $S_{\mathcal{X}}^{-1/2}$ satisfies $[S_{\mathcal{X}}^{-1/2}]_{00} = [S_{\mathcal{X}}^{-1/2}]_{11}$ (by its spectral decomposition, since the eigenvectors of Eq.~\eqref{eq:frame-op-d2-coord} have equal-magnitude components), the resulting L\"owdin vector $|v_{\mathcal{X}}^0\rangle = S_{\mathcal{X}}^{-1/2}|h_0\rangle$ has the conjugate-symmetric form
\begin{align}
    |v_{\mathcal{X}}^0\rangle = \frac{1}{\sqrt{2}}\begin{pmatrix} e^{i\theta_{\mathcal{X}}} \\ e^{-i\theta_{\mathcal{X}}} \end{pmatrix}
    \label{eq:lowdin-conj-sym}
\end{align}
for some real angle $\theta_{\mathcal{X}}$ depending on $m_1$.
The eigenvalues of $S_{\mathcal{X}}$ are $\lambda_{\pm} = 1 \pm r_1$ (both positive since $1 + r_1 = 2\cos^2(\pi m_1/(2n)) > 0$ and $1 - r_1 = 2\sin^2(\pi m_1/(2n)) > 0$ for $m_1 \in \{1, \ldots, n-1\}$), with eigenvectors $|u_\pm\rangle = \frac{1}{\sqrt{2}}(1, \pm e^{i\phi_1})^T$. The inverse square root is $S_{\mathcal{X}}^{-1/2} = \sum_{\pm} (1 \pm r_1)^{-1/2} |u_\pm\rangle\langle u_\pm|$. The inner products with $|h_0\rangle = \frac{1}{\sqrt{2}}(1,1)^T$ are
\begin{align}
    \langle u_\pm | h_0 \rangle = \tfrac{1}{2}(1 \pm e^{-i\phi_1}),
    \label{eq:u-h0-overlap}
\end{align}
so the first component of $|v_{\mathcal{X}}^0\rangle = S_{\mathcal{X}}^{-1/2}|h_0\rangle$ evaluates to
\begin{align}
    [v_{\mathcal{X}}^0]_0
    &= \frac{1}{2\sqrt{2}}\left(\frac{1+e^{-i\phi_1}}{\sqrt{1+r_1}} + \frac{1-e^{-i\phi_1}}{\sqrt{1-r_1}}\right).
    \label{eq:v0-first-component}
\end{align}
Using the half-angle identities
\begin{align}
    1 + r_1 &= 1 + \cos\!\left(\frac{\pi m_1}{n}\right) = 2\cos^2\!\left(\frac{\pi m_1}{2n}\right), \label{eq:half-angle-plus}\\
    1 - r_1 &= 1 - \cos\!\left(\frac{\pi m_1}{n}\right) = 2\sin^2\!\left(\frac{\pi m_1}{2n}\right), \label{eq:half-angle-minus}
\end{align}
together with $1 + e^{-i\phi_1} = 2\cos(\phi_1/2)\,e^{-i\phi_1/2}$ and $1 - e^{-i\phi_1} = 2i\sin(\phi_1/2)\,e^{-i\phi_1/2}$, Eq.~\eqref{eq:v0-first-component} simplifies to
\begin{align}
    [v_{\mathcal{X}}^0]_0
    = \frac{e^{-i\phi_1/2}}{2\sqrt{2}}\left(\sqrt{2} + i\sqrt{2}\right)
    = \frac{1}{\sqrt{2}}\,e^{i(\pi/4 - \phi_1/2)}.
    \label{eq:v0-first-simplified}
\end{align}
An analogous computation gives the second component
\begin{align}
    [v_{\mathcal{X}}^0]_1
    &= \frac{1}{2\sqrt{2}}\left(\frac{e^{i\phi_1}+1}{\sqrt{1+r_1}} - \frac{e^{i\phi_1}-1}{\sqrt{1-r_1}}\right)
    = \frac{1}{\sqrt{2}}\,e^{-i(\pi/4 - \phi_1/2)},
    \label{eq:v0-second-component}
\end{align}
confirming the conjugate-symmetric form $[v_{\mathcal{X}}^0]_1 = \overline{[v_{\mathcal{X}}^0]_0}$ of Eq.~\eqref{eq:lowdin-conj-sym}. Identifying the phase $[v_{\mathcal{X}}^0]_0 = \frac{1}{\sqrt{2}}e^{i\theta_{\mathcal{X}}}$ yields
\begin{align}
    \theta_{\mathcal{X}} = \frac{\pi}{4} - \frac{\pi m_1}{2n}.
    \label{eq:theta-X}
\end{align}
The same derivation for $\mathcal{Y} = \{0, m_2\}$ gives $\theta_{\mathcal{Y}} = \frac{\pi}{4} - \frac{\pi m_2}{2n}$.

For two vectors of the conjugate-symmetric form Eq.~\eqref{eq:lowdin-conj-sym}, we can then evaluate the inner product as
\begin{align}
    \langle v_{\mathcal{X}}^0 | v_{\mathcal{Y}}^0 \rangle
    = \frac{1}{2}\left(e^{-i\theta_{\mathcal{X}}} e^{i\theta_{\mathcal{Y}}} + e^{i\theta_{\mathcal{X}}} e^{-i\theta_{\mathcal{Y}}}\right)
    = \cos(\theta_{\mathcal{Y}} - \theta_{\mathcal{X}}).
    \label{eq:overlap-cosine}
\end{align}
Taking the squared modulus and substituting $\theta_{\mathcal{X}}$ and $\theta_{\mathcal{Y}}$ from Eq.~\eqref{eq:theta-X} we obtain
\begin{align}
    |\langle v_{\mathcal{X}}^0 | v_{\mathcal{Y}}^0 \rangle|^2
    = \cos^2(\theta_{\mathcal{Y}} - \theta_{\mathcal{X}})
    = \cos^2\!\left(\frac{\pi(m_1 - m_2)}{2n}\right)
    = \cos^2\!\left(\frac{\pi|m_2 - m_1|}{2n}\right),
    \label{eq:overlap-final}
\end{align}
where the last equality uses $\cos^2(-x) = \cos^2(x)$. This establishes Eq.~\eqref{eq:cross-overlap-d2}.
\end{proof}

We are now ready to put everything together to establish the quantum winning probability $\omega_q$ stated in Theorem~\ref{thm:harmonic-d2}.
\begin{proof}[Proof of Theorem~\ref{thm:harmonic-d2}]
	From~\eqref{eq:inTermsOfA} the quantum winning probability is
\begin{align}
    \omega_q = \frac{1}{d|\mathcal{S}|^2} \sum_{c \in [n]} \tr[A_c^2],
    \label{eq:omega-q-proof}
\end{align}
where $A_c = \sum_{\mathcal{X} \ni c} \ket{v_{\mathcal{X}}^c}\bra{v_{\mathcal{X}}^c}$ and $|\mathcal{S}| = \binom{n}{2} = n(n-1)/2$ for $d = 2$.
For each channel $c$, the safe sets containing $c$ are precisely the pairs $\{c, j\}$ for $j \neq c$, giving $n - 1$ such sets. We can now expand
\begin{align}
    \tr[A_c^2] = \sum_{j, k \neq c} |\inp{v_{\{c,j\}}^c}{v_{\{c,k\}}^c}|^2.
    \label{eq:trace-expand}
\end{align}
The diagonal terms ($j = k$) contribute $(n-1)$ terms of value $1$. For the off-diagonal terms, by Lemma~\ref{lem:cross-overlap-d2}, the overlap depends only on the difference $\Delta = |m_2 - m_1|$ where $m_i = (i - c) \mod n$. For each $\Delta \in \{1, \ldots, n-2\}$, there are exactly $2(n - 1 - \Delta)$ ordered pairs $(j, k)$ with $|(j-c) - (k-c)| \equiv \Delta \pmod{n}$.
We thus obtain
\begin{align}
    \tr[A_c^2] &= (n-1) + 2\sum_{\Delta=1}^{n-2} (n - 1 - \Delta) \cos^2\left(\frac{\pi\Delta}{2n}\right).
    \label{eq:trace-sum}
\end{align}
This expression is independent of $c$ by cyclic symmetry, so summing over all $n$ channels gives us
\begin{align}
    \sum_{c \in [n]} \tr[A_c^2] = n(n-1) + 2n \sum_{\Delta=1}^{n-2} (n-1-\Delta) \cos^2\left(\frac{\pi\Delta}{2n}\right).
    \label{eq:sum-channels}
\end{align}
Using the identity $\cos^2(x) = (1 + \cos(2x))/2$, we can write
\begin{align}
    \sum_{\Delta=1}^{n-2} (n-1-\Delta) \cos^2\left(\frac{\pi\Delta}{2n}\right) 
    = \frac{(n-2)(n-1)}{4} + \frac{T_n}{2},
    \label{eq:cos2-identity}
\end{align}
where $T_n$ is defined in Eq.~\eqref{eq:T-n-def}.
Substituting Eq.~\eqref{eq:cos2-identity} into Eq.~\eqref{eq:sum-channels} and then into Eq.~\eqref{eq:omega-q-proof} with $d = 2$ and $|\mathcal{S}|^2 = n^2(n-1)^2/4$ gives
\begin{align}
    \omega_q &= \frac{2}{n^2(n-1)^2} \left[ n(n-1) + \frac{n(n-2)(n-1)}{2} + nT_n \right] \notag \\
    &= \frac{2}{n(n-1)^2} \left[ (n-1) + \frac{(n-2)(n-1)}{2} + T_n \right] \notag \\
    &= \frac{2}{n(n-1)^2} \left[ \frac{n(n-1)}{2} + T_n \right] \notag \\
    &= \frac{1}{n-1} + \frac{2T_n}{n(n-1)^2}. 
\end{align}
\end{proof}

\begin{remark}[Connection to simplex frame]
For $n = 3$, the harmonic frame in $\Complex^2$ consists of three equiangular vectors with inner product magnitude $|\cos(\pi/3)| = 1/2$, which is precisely a regular simplex. The value $\omega_q^{\mathrm{harm}}(3) = 7/12$ matches the simplex result from Section~\ref{sec:simplex}.
\end{remark}

\subsubsection{Quantum Advantage}\label{sec:harmonic-advantage}

We now establish that the harmonic strategy with $d = 2$ achieves quantum advantage for all $n \geq 3$. The key ingredient is a closed-form evaluation of $T_n$, which reduces the advantage condition to an elementary inequality.

\begin{theorem}[Quantum advantage for harmonic $d = 2$]\label{thm:harmonic-advantage}
For all $n \geq 3$, the $(n, n-2)$-jamming game with harmonic seed frame satisfies
\begin{align}
    \omega_q^{\mathrm{harm}}(n) > \omega_c(n, n-2).
    \label{eq:harmonic-advantage}
\end{align}
Moreover, the quantum-to-classical ratio has the asymptotic limit
\begin{align}
    \lim_{n \to \infty} \frac{\omega_q^{\mathrm{harm}}(n)}{\omega_c(n, n-2)} = \frac{3}{4} + \frac{3}{\pi^2} \approx 1.054.
    \label{eq:harmonic-ratio-limit}
\end{align}
\end{theorem}

The proof proceeds in four steps. We first simplify the classical value $\omega_c(n, n-2)$. Second, we derive a closed-form expression for $T_n$ using character sum identities. Third, we reformulate the quantum advantage condition as an inequality on $T_n$. Finally, we establish this inequality analytically for all $n \geq 3$.

\begin{corollary}[Classical value for $d = 2$]\label{lem:classical-d2}
The classical value of the $(n, n-2)$-jamming game is
\begin{align}
    \omega_c(n, n-2) = \frac{2(2n-1)}{3n(n-1)}.
    \label{eq:classical-d2}
\end{align}
\end{corollary}
\begin{proof}
From Theorem~\ref{thm:classical-value} with $k = n - 2$:
\begin{align}
    \omega_c(n, n-2) = \frac{1}{\binom{n}{n-2}^2} \sum_{i=0}^{n-2} \binom{n-1-i}{n-2-i}^2.
    \label{eq:classical-from-thm}
\end{align}
Now $\binom{n}{n-2} = \binom{n}{2} = n(n-1)/2$, and $\binom{n-1-i}{n-2-i} = n - 1 - i$. The sum becomes
\begin{align}
    \sum_{i=0}^{n-2} (n-1-i)^2 = \sum_{j=1}^{n-1} j^2 = \frac{(n-1)n(2n-1)}{6},
    \label{eq:sum-squares}
\end{align}
where the first equality uses the substitution $j = n - 1 - i$, and the second is the standard formula for the sum of squares~\cite{GrahamKnuthPatashnik}. Substituting into Eq.~\eqref{eq:classical-from-thm}:
\begin{align}
    \omega_c(n, n-2) = \frac{(n-1)n(2n-1)/6}{[n(n-1)/2]^2} = \frac{(n-1)n(2n-1)/6}{n^2(n-1)^2/4} = \frac{2(2n-1)}{3n(n-1)}. 
\end{align}
\end{proof}

The next lemma provides an exact closed-form expression for $T_n$ that will be the basis of our analytical proof.

\begin{lemma}[Closed form for $T_n$]\label{lem:T-closed-form}
The sum $T_n = \sum_{\Delta=1}^{n-2} (n - 1 - \Delta)\cos(\pi\Delta/n)$ admits the closed-form expression
\begin{align}
    T_n = \frac{1}{2\sin^2(\pi/(2n))} - \frac{n}{2}.
    \label{eq:T-closed-form}
\end{align}
\end{lemma}
\begin{proof}
We split $T_n = (n-1)C_1 - C_2$, where
\begin{align}
    C_1 := \sum_{\Delta=1}^{n-2} \cos\!\left(\frac{\pi\Delta}{n}\right), \qquad
    C_2 := \sum_{\Delta=1}^{n-2} \Delta \cos\!\left(\frac{\pi\Delta}{n}\right).
    \label{eq:C1-C2-def}
\end{align}
We evaluate each using the complex exponential $\tau := e^{i\pi/n}$, noting that $\tau^n = e^{i\pi} = -1$.

Evaluation of $C_1$: The complete geometric sum satisfies
\begin{align}
    \sum_{\Delta=0}^{n-1} \tau^\Delta = \frac{1 - \tau^n}{1 - \tau} = \frac{2}{1 - \tau}.
    \label{eq:geometric-sum}
\end{align}
Writing $1 - \tau = -2i\sin(\pi/(2n))\,e^{i\pi/(2n)}$, one finds $\operatorname{Re}\!\left[\frac{2}{1-\tau}\right] = 1$, so $\sum_{\Delta=0}^{n-1} \cos(\pi\Delta/n) = 1$. Removing the $\Delta = 0$ term (which equals $1$) and the $\Delta = n-1$ term (which equals $\cos(\pi(n-1)/n) = -\cos(\pi/n)$):
\begin{align}
    C_1 = 1 - 1 + \cos(\pi/n) = \cos(\pi/n).
    \label{eq:C1-result}
\end{align}

Evaluation of $C_2$: Differentiating the geometric sum $\sum_{\Delta=0}^{N-1} z^\Delta = (1 - z^N)/(1 - z)$ with respect to $z$ and multiplying by $z$ gives the standard identity (see e.g.~\cite[Eq.~(2.31)]{GrahamKnuthPatashnik})
\begin{align}
    \sum_{\Delta=0}^{N-1} \Delta\, z^\Delta = \frac{z\bigl[1 - Nz^{N-1} + (N-1)z^N\bigr]}{(1-z)^2}.
    \label{eq:weighted-geometric}
\end{align}
Setting $N = n$ and $z = \tau = e^{i\pi/n}$, we use $\tau^n = -1$ and $\tau^{n-1} = -\tau^{-1} = -e^{-i\pi/n}$ to simplify the numerator:
\begin{align}
    \tau\bigl[1 + ne^{-i\pi/n} - (n-1)\bigr]
    &= \tau\bigl[2 - n + ne^{-i\pi/n}\bigr]
    = \tau\bigl[2 - n(1 - e^{-i\pi/n})\bigr].
    \label{eq:numerator-simplify}
\end{align}
Since $\tau\,(1 - \bar\tau) = \tau - |\tau|^2 = \tau - 1 = -(1 - \tau)$, this becomes $2\tau + n(1-\tau)$, giving
\begin{align}
    \sum_{\Delta=0}^{n-1} \Delta\,\tau^\Delta = \frac{2\tau}{(1-\tau)^2} + \frac{n}{1-\tau}.
    \label{eq:weighted-sum-split}
\end{align}
For the first term, using $(1-\tau)^2 = -4\sin^2(\pi/(2n))\,e^{i\pi/n}$:
\begin{align}
    \frac{2\tau}{(1-\tau)^2} = \frac{2e^{i\pi/n}}{-4\sin^2(\pi/(2n))\,e^{i\pi/n}} = \frac{-1}{2\sin^2(\pi/(2n))},
    \label{eq:first-term-real}
\end{align}
which is real. For the second term, a direct calculation gives $\operatorname{Re}\!\left[\frac{1}{1-\tau}\right] = \frac{1}{2}$ (see the evaluation of $C_1$). Combining:
\begin{align}
    \sum_{\Delta=0}^{n-1} \Delta\cos\!\left(\frac{\pi\Delta}{n}\right) = \frac{-1}{2\sin^2(\pi/(2n))} + \frac{n}{2}.
    \label{eq:weighted-cos-full}
\end{align}
Since the $\Delta = 0$ term vanishes, removing only the $\Delta = n-1$ term yields
\begin{align}
    C_2 = \frac{-1}{2\sin^2(\pi/(2n))} + \frac{n}{2} - (n-1)\cos\!\left(\frac{\pi(n-1)}{n}\right) = \frac{-1}{2\sin^2(\pi/(2n))} + \frac{n}{2} + (n-1)\cos\!\left(\frac{\pi}{n}\right),
    \label{eq:C2-result}
\end{align}
where we used $\cos(\pi(n-1)/n) = -\cos(\pi/n)$.

Combining: Substituting Eqs.~\eqref{eq:C1-result} and~\eqref{eq:C2-result} into $T_n = (n-1)C_1 - C_2$:
\begin{align}
    T_n &= (n-1)\cos(\pi/n) - \left[\frac{-1}{2\sin^2(\pi/(2n))} + \frac{n}{2} + (n-1)\cos(\pi/n)\right] \notag \\
    &= \frac{1}{2\sin^2(\pi/(2n))} - \frac{n}{2}. 
\end{align}
\end{proof}

It is now convenient to rewrite the condition for obtaining a quantum advantage in terms of $T_n$.
\begin{lemma}[Equivalent condition for quantum advantage]\label{lem:advantage-equiv}
For $n \geq 3$, the inequality $\omega_q^{\mathrm{harm}}(n) > \omega_c(n, n-2)$ holds if and only if
\begin{align}
    T_n > \frac{(n-1)(n-2)}{6}
    \label{eq:T-threshold}
\end{align}
\end{lemma}
\begin{proof}
Using Theorem~\ref{thm:harmonic-d2} and Corollary~\ref{lem:classical-d2}, the condition $\omega_q > \omega_c$ becomes
\begin{align}
    \frac{1}{n-1} + \frac{2T_n}{n(n-1)^2} > \frac{2(2n-1)}{3n(n-1)}.
\end{align}
Multiplying both sides by $n(n-1)^2$:
\begin{align}
    n(n-1) + 2T_n > \frac{2(2n-1)(n-1)}{3}.
\end{align}
Rearranging:
\begin{align}
    2T_n > \frac{2(2n-1)(n-1)}{3} - n(n-1) = \frac{(n-1)\bigl[2(2n-1) - 3n\bigr]}{3} = \frac{(n-1)(n-2)}{3}.
\end{align}
Dividing by $2$ yields Eq.~\eqref{eq:T-threshold}.
\end{proof}

We are now ready to give a fully analytical proof of Theorem~\ref{thm:harmonic-advantage}.

\begin{proof}[Proof of Theorem~\ref{thm:harmonic-advantage}]
By Lemma~\ref{lem:advantage-equiv}, quantum advantage holds if and only if $T_n > (n-1)(n-2)/6$. Using the closed form from Lemma~\ref{lem:T-closed-form}, this is equivalent to
\begin{align}
    \frac{1}{2\sin^2(\pi/(2n))} - \frac{n}{2} > \frac{(n-1)(n-2)}{6},
    \label{eq:advantage-closed}
\end{align}
which rearranges to
\begin{align}
    \frac{1}{2\sin^2(\pi/(2n))} > \frac{n}{2} + \frac{(n-1)(n-2)}{6} = \frac{3n + n^2 - 3n + 2}{6} = \frac{n^2 + 2}{6}.
    \label{eq:advantage-rearranged}
\end{align}

Case $n \geq 4$: For $x > 0$, the elementary inequality $\sin x < x$ gives $\sin^2(\pi/(2n)) < \pi^2/(4n^2)$, and hence
\begin{align}
    \frac{1}{2\sin^2(\pi/(2n))} > \frac{2n^2}{\pi^2}.
    \label{eq:sin-lower-bound}
\end{align}
Thus Eq.~\eqref{eq:advantage-rearranged} holds whenever $2n^2/\pi^2 > (n^2 + 2)/6$, i.e.,
\begin{align}
    12n^2 > \pi^2(n^2 + 2) \quad \iff \quad n^2(12 - \pi^2) > 2\pi^2 \quad \iff \quad n^2 > \frac{2\pi^2}{12 - \pi^2}.
    \label{eq:n-threshold}
\end{align}
Since $12 - \pi^2 > 2.13$, the right-hand side satisfies $2\pi^2/(12 - \pi^2) < 2 \times 9.87/2.13 < 9.27$. Therefore Eq.~\eqref{eq:n-threshold} holds for all $n \geq 4$.

Case $n = 3$: By Lemma~\ref{lem:T-closed-form}, $T_3 = 1/(2\sin^2(\pi/6)) - 3/2 = 1/(2 \cdot 1/4) - 3/2 = 1/2$. The threshold is $(2 \cdot 1)/6 = 1/3$. Since $1/2 > 1/3$, Eq.~\eqref{eq:T-threshold} holds for $n = 3$.

Asymptotic ratio: For the limiting ratio, we use $\sin(\pi/(2n)) = \pi/(2n) + O(n^{-3})$ to expand the quantum value via Theorem~\ref{thm:harmonic-d2} and Lemma~\ref{lem:T-closed-form}:
\begin{align}
    T_n = \frac{2n^2}{\pi^2} - \frac{n}{2} + O(1), \qquad \omega_q = \frac{1}{n}\left(1 + \frac{4}{\pi^2}\right) + O(n^{-2}).
\end{align}
Since $\omega_c = 4/(3n) + O(n^{-2})$ by Corollary~\ref{lem:classical-d2}:
\begin{align}
    \frac{\omega_q}{\omega_c} \to \frac{1 + 4/\pi^2}{4/3} = \frac{3}{4} + \frac{3}{\pi^2} \approx 1.054. 
\end{align}
\end{proof}

\subsubsection{\texorpdfstring{Numerical evaluation for general $d$}{Numerical evaluation for general d}}

For $d \geq 3$, the cross-context overlaps $|\inp{v_{\mathcal{X}}^c}{v_{\mathcal{Y}}^c}|^2$ no longer have the simple closed form of Lemma~\ref{lem:cross-overlap-d2}, and we evaluate $\omega_q$ by direct numerical computation via Eq.~\eqref{eq:inTermsOfA} using our numerical package~\cite{JammingNumerics}. Table~\ref{tab:harmonic-advantage} reports selected results for $d = 2$ through $d = 6$.

\begin{table}[ht]
\centering
\begin{tabular}{cc|ccc}
\toprule
$d$ & $n$ & $\omega_q^{\mathrm{harm}}$ & $\omega_c$ & $\omega_q/\omega_c$ \\
\midrule
2 & 3 & 0.5833 & 0.5556 & 1.050 \\
2 & 6 & 0.2595 & 0.2444 & 1.062 \\
2 & 10 & 0.1492 & 0.1407 & 1.060 \\
2 & 20 & 0.0724 & 0.0684 & 1.058 \\
\midrule
3 & 4 & 0.6451 & 0.6250 & 1.032 \\
3 & 6 & 0.3783 & 0.3650 & 1.036 \\
3 & 12 & 0.1692 & 0.1641 & 1.031 \\
3 & 20 & 0.0974 & 0.0948 & 1.028 \\
\midrule
4 & 5 & 0.6909 & 0.6800 & 1.016 \\
4 & 6 & 0.5290 & 0.5200 & 1.017 \\
4 & 12 & 0.2211 & 0.2186 & 1.012 \\
4 & 20 & 0.1246 & 0.1237 & 1.007 \\
\midrule
5 & 6 & 0.7260 & 0.7222 & 1.005 \\
5 & 10 & 0.3508 & 0.3504 & 1.001 \\
5 & 15 & 0.2137 & 0.2144 & 0.997 \\
5 & 20 & 0.1537 & 0.1546 & 0.994 \\
\midrule
6 & 7 & 0.7537 & 0.7551 & 0.998 \\
6 & 10 & 0.4388 & 0.4420 & 0.993 \\
6 & 15 & 0.2599 & 0.2630 & 0.988 \\
6 & 20 & 0.1847 & 0.1874 & 0.985 \\
\bottomrule
\end{tabular}
\caption{Quantum winning probability for the harmonic strategy across different game parameters. The $d = 2$ values agree with the analytical formula of Theorem~\ref{thm:harmonic-d2}. The advantage ratio is persistent for $d \leq 4$, crosses over for $d = 5$, and is absent for $d = 6$.}
\label{tab:harmonic-advantage}
\end{table}

Several patterns emerge beyond the $d = 2$ case analyzed above. For $d = 3$ and $d = 4$, the harmonic strategy achieves $\omega_q > \omega_c$ for all values of $n$ tested (up to $n = 20$), and the advantage ratio stabilizes as $n$ increases: approximately $1.028$ for $d = 3$ and $1.007$ for $d = 4$. Like the simplex (Proposition~\ref{prop:simplex-crossover}), the harmonic strategy exhibits a crossover for $d = 5$, where the advantage disappears near $n \approx 11$. At $d = 6$, the harmonic strategy yields $\omega_q < \omega_c$ for all $n$ tested.

Comparing with the random frame predictions of Theorem~\ref{thm:random-winning}, the asymptotic advantage ratios $(2d-1)\mathcal{L}_1/d \approx 1.083, 1.075, 1.069$ for $d = 2, 3, 4$ exceed those of the harmonic frame. This is consistent with the fact that harmonic vectors are confined to a specific algebraic structure---the discrete Fourier subspace---while random frames exploit the full state space geometry. Nevertheless, the harmonic frame provides a concrete, deterministic strategy achieving quantum advantage for all $d \leq 4$ regardless of $n$.

%% file: simplexStrategy.tex
\subsection{Simplex Strategy}\label{sec:simplex}

We now consider seed vectors that form a regular simplex, the simplest example of an equiangular tight frame~\cite{strohmer2003grassmannian,tropp2005complex}.

\subsubsection{Construction}
A \emph{regular simplex} in $\mathbb{C}^d$ consists of $n = d+1$ unit vectors with constant pairwise inner product magnitude:
\begin{align}
    |\langle s_i | s_j \rangle|^2 = \frac{1}{d^2} \quad \text{for all } i \neq j.
    \label{eq:simplex_condition}
\end{align}
This means using the simplex corresponds to a jamming game with $n = d+1$ channels where the adversary blocks only $k = 1$ channels, leaving safe sets of size $|x| = d$. While limited to this specific regime, the simplex seed frame has the advantage of existing in all dimensions (unlike MUBs and SIC-POVMs)~\cite{fickus2018equiangular} and admits closed-form analysis.
An explicit construction projects the $d+1$ standard basis vectors of $\mathbb{R}^{d+1}$ onto the orthogonal complement of the all-ones vector and normalizes:
\begin{align}
    |s_j\rangle = \sqrt{\frac{d+1}{d}} \left( |e_j\rangle - \frac{1}{d+1} \sum_{k=0}^{d} |e_k\rangle \right), \quad j \in \{0, \ldots, d\}.
    \label{eq:simplex_construction}
\end{align}
The coherence is
\begin{align}
    \delta_{\mathrm{simplex}} = \frac{1}{d},
    \label{eq:simplex_coherence}
\end{align}
which achieves the Welch bound~\cite{welch1974} for $d+1$ vectors in $\mathbb{C}^d$.

\subsubsection{Exact Analysis}

We now derive an exact analytical expression for the quantum winning probability using simplex seed vectors.

\begin{theorem}[Simplex quantum value]\label{thm:simplex-exact}
For the $(d+1, 1)$-jamming game using the simplex seed frame, the quantum winning probability is
\begin{align}
    \omega_q^{\mathrm{simplex}} = \frac{1}{d+1} + \frac{d-1}{d+1} \cdot \mu^2,
\end{align}
where the cross-context overlap $\mu$ is given by
\begin{align}
    \mu = \frac{d^3 - 3d - 2 + 2(d+1)^{3/2}}{d^2(d+1)}.
\end{align}
\end{theorem}

The proof proceeds in four steps: we first establish the structure of the Gram matrix, then derive a closed form for the L\"owdin vectors, use this to compute the operator $A_c$, and finally evaluate the trace.

\begin{lemma}[Gram matrix structure]\label{lem:simplex-gram}
For the simplex seed frame $\{|s_c\rangle\}_{c=0}^d$ in $\mathbb{C}^d$, the Gram matrix of any $d$-subset $X$ takes the form
\begin{align}
    G_X = \frac{d+1}{d} I - \frac{1}{d} \mathbf{1}\mathbf{1}^T, \label{eq:simplex-gram}
\end{align}
where $\mathbf{1} = (1, \ldots, 1)^T \in \mathbb{R}^d$. The eigenvalues are $\frac{d+1}{d}$ with multiplicity $d-1$ and $\frac{1}{d}$ with multiplicity $1$.
\end{lemma}

\begin{proof}
The simplex vectors satisfy $\langle s_c | s_c \rangle = 1$ and $\langle s_c | s_{c'} \rangle = -1/d$ for $c \neq c'$. Thus $[G_X]_{c,c'} = \delta_{c,c'} - (1-\delta_{c,c'})/d = (1 + 1/d)\delta_{c,c'} - 1/d$, which can be written as Eq.~\eqref{eq:simplex-gram}.

For the eigenvalues, note that $\mathbf{1}$ is an eigenvector of $\mathbf{1}\mathbf{1}^T$ with eigenvalue $d$, and vectors orthogonal to $\mathbf{1}$ have eigenvalue $0$. Therefore $G_X$ has eigenvalue $\frac{d+1}{d} - \frac{d}{d} = \frac{1}{d}$ for the $\mathbf{1}$ direction, and eigenvalue $\frac{d+1}{d}$ with multiplicity $d-1$ for the orthogonal complement.
\end{proof}

\begin{lemma}[L\"owdin vectors for simplex]\label{lem:simplex-lowdin}
Let $x_a = [n] \setminus \{a\}$ denote the safe set excluding channel $a$. For any $c \in x_a$ (i.e., $c \neq a$), the L\"owdin orthonormalized vector is
\begin{align}
    |v_c^{x_a}\rangle = \alpha |s_c\rangle - \gamma |s_a\rangle, \label{eq:simplex-lowdin-form}
\end{align}
where
\begin{align}
    \alpha = \sqrt{\frac{d}{d+1}}, \qquad \gamma = \frac{\sqrt{d+1} - 1}{\sqrt{d(d+1)}}. \label{eq:alpha-gamma-def}
\end{align}
\end{lemma}

\begin{proof}
From Lemma~\ref{lem:simplex-gram}, the inverse square root of the Gram matrix is
\begin{align}
    G_{x_a}^{-1/2} = \sqrt{\frac{d}{d+1}} I + \frac{1}{d}\left(\sqrt{d} - \sqrt{\frac{d}{d+1}}\right) \mathbf{1}\mathbf{1}^T. \label{eq:gram-inv-sqrt}
\end{align}
This follows from the spectral decomposition: if $P = \mathbf{1}\mathbf{1}^T/d$ projects onto the all-ones direction, then $G_X = \frac{d+1}{d}(I - P) + \frac{1}{d}P$, and taking the inverse square root of each eigenvalue gives Eq.~\eqref{eq:gram-inv-sqrt}.

The matrix elements are $[G_{x_a}^{-1/2}]_{c,c'} = \alpha \delta_{c,c'} + \gamma$ where $\alpha = \sqrt{d/(d+1)}$ and $\gamma = \frac{1}{d}(\sqrt{d} - \sqrt{d/(d+1)})$. By Definition~\ref{def:lowdin}, the L\"owdin vector is
\begin{align}
    |v_c^{x_a}\rangle &= \sum_{c' \in x_a} [G_{x_a}^{-1/2}]_{c',c} |s_{c'}\rangle \\
    &= \alpha |s_c\rangle + \gamma \sum_{c' \neq a} |s_{c'}\rangle.
\end{align}
The simplex vectors satisfy $\sum_{j=0}^d |s_j\rangle = 0$ (since they form a centered regular simplex), so $\sum_{c' \neq a} |s_{c'}\rangle = -|s_a\rangle$. Simplifying yields Eq.~\eqref{eq:simplex-lowdin-form}.

To verify Eq.~\eqref{eq:alpha-gamma-def}, we compute
\begin{align}
    \gamma = \frac{1}{d}\left(\sqrt{d} - \sqrt{\frac{d}{d+1}}\right) = \frac{1}{\sqrt{d}}\left(1 - \frac{1}{\sqrt{d+1}}\right) = \frac{\sqrt{d+1} - 1}{\sqrt{d(d+1)}}.
\end{align}
\end{proof}

The elegant form of Eq.~\eqref{eq:simplex-lowdin-form} is the key to the exact analysis: each L\"owdin vector is a linear combination of just two simplex vectors---the one for channel $c$, and the one for the excluded channel $a$.

\begin{lemma}[Cross-context overlap]\label{lem:simplex-overlap}
For distinct safe sets $x_a$ and $x_b$ with $a \neq b$, and any channel $c \in x_a \cap x_b$ (i.e., $c \notin \{a, b\}$), the overlap between L\"owdin vectors is
\begin{align}
    \mu \;:=\; \langle v_c^{x_a} | v_c^{x_b} \rangle = \alpha^2 + \frac{2\alpha\gamma}{d} - \frac{\gamma^2}{d}. \label{eq:mu-def}
\end{align}
This overlap depends only on the dimension $d$, not on the specific channels $a$, $b$, $c$.
\end{lemma}

\begin{proof}
Using Eq.~\eqref{eq:simplex-lowdin-form}:
\begin{align}
    \langle v_c^{x_a} | v_c^{x_b} \rangle &= (\alpha \langle s_c| - \gamma \langle s_a|)(\alpha |s_c\rangle - \gamma |s_b\rangle) \\
    &= \alpha^2 \langle s_c|s_c\rangle - \alpha\gamma \langle s_c|s_b\rangle - \alpha\gamma \langle s_a|s_c\rangle + \gamma^2 \langle s_a|s_b\rangle.
\end{align}
Since $c \notin \{a, b\}$ and $a \neq b$, all pairs are distinct, so $\langle s_c|s_c\rangle = 1$ and $\langle s_i|s_j\rangle = -1/d$ for the remaining three inner products. Substituting:
\begin{align}
    \langle v_c^{x_a} | v_c^{x_b} \rangle = \alpha^2 - \alpha\gamma\left(-\frac{1}{d}\right) - \alpha\gamma\left(-\frac{1}{d}\right) + \gamma^2\left(-\frac{1}{d}\right) = \alpha^2 + \frac{2\alpha\gamma}{d} - \frac{\gamma^2}{d}.
\end{align}
\end{proof}

\begin{lemma}[Operator structure]\label{lem:simplex-Ac}
For any channel $c \in [n]$, the operator $A_c = \sum_{a \neq c} |v_c^{x_a}\rangle\langle v_c^{x_a}|$ satisfies
\begin{align}
    \mathrm{tr}[A_c^2] = d + d(d-1)\mu^2, \label{eq:trace-Ac-squared}
\end{align}
where $\mu$ is defined in Eq.~\eqref{eq:mu-def}.
\end{lemma}

\begin{proof}
Expanding the trace:
\begin{align}
    \mathrm{tr}[A_c^2] = \sum_{a, b \neq c} |\langle v_c^{x_a} | v_c^{x_b} \rangle|^2.
\end{align}
This sum has $d^2$ terms (since there are $d$ choices for each of $a$ and $b$ from $[n] \setminus \{c\}$). When $a = b$, we have $|\langle v_c^{x_a} | v_c^{x_a} \rangle|^2 = 1$ by orthonormality, contributing $d$ terms. When $a \neq b$, Lemma~\ref{lem:simplex-overlap} gives $|\langle v_c^{x_a} | v_c^{x_b} \rangle|^2 = \mu^2$, contributing $d(d-1)$ terms. Thus Eq.~\eqref{eq:trace-Ac-squared} follows.
\end{proof}

\begin{proof}[Proof of Theorem~\ref{thm:simplex-exact}]
From Eq.~\eqref{eq:omega-from-Ac}, the quantum winning probability is
\begin{align}
    \omega_q = \frac{1}{d|S|^2} \sum_{c \in [n]} \mathrm{tr}[A_c^2].
\end{align}
For the simplex, $|S| = \binom{d+1}{1} = d+1$ and $n = d+1$. By symmetry, $\mathrm{tr}[A_c^2]$ is independent of $c$, so using Lemma~\ref{lem:simplex-Ac}:
\begin{align}
    \omega_q &= \frac{(d+1) \cdot (d + d(d-1)\mu^2)}{d(d+1)^2} = \frac{d(1 + (d-1)\mu^2)}{d(d+1)} = \frac{1 + (d-1)\mu^2}{d+1}.
\end{align}

It remains to express $\mu$ explicitly. Using Eq.~\eqref{eq:alpha-gamma-def}:
\begin{align}
    \alpha^2 &= \frac{d}{d+1}, \\
    \alpha\gamma &= \sqrt{\frac{d}{d+1}} \cdot \frac{\sqrt{d+1} - 1}{\sqrt{d(d+1)}} = \frac{\sqrt{d+1} - 1}{d+1}, \\
    \gamma^2 &= \frac{(\sqrt{d+1} - 1)^2}{d(d+1)} = \frac{d + 2 - 2\sqrt{d+1}}{d(d+1)}.
\end{align}
Substituting into Eq.~\eqref{eq:mu-def} and simplifying:
\begin{align}
    \mu = \frac{d^3 - 3d - 2 + 2(d+1)^{3/2}}{d^2(d+1)}.
\end{align}
\end{proof}

\subsubsection{Asymptotic behavior and the crossover phenomenon}

A natural question is how the quantum-to-classical ratio $\omega_q^{\mathrm{simplex}}/\omega_c$ behaves as the dimension increases. Remarkably, the simplex strategy exhibits a \emph{crossover}: it provides quantum advantage only for small dimensions.

\begin{proposition}[Crossover dimension]\label{prop:simplex-crossover}
The simplex strategy satisfies $\omega_q^{\mathrm{simplex}} > \omega_c$ if and only if $d \leq 5$. The crossover occurs at $d^* \approx 5.70$, where $\omega_q^{\mathrm{simplex}} = \omega_c$.
\end{proposition}

\begin{proof}
The condition $\omega_q^{\mathrm{simplex}} > \omega_c$ reduces to $(d+1)\mu^2 > d$, which we verify by direct computation for small $d$ and asymptotic analysis for large $d$.

\emph{Small dimensions:} Substituting the exact formula $\mu = (d^3 - 3d - 2 + 2(d+1)^{3/2})/(d^2(d+1))$ and evaluating $(d+1)\mu^2 - d$ for $d = 2, \ldots, 7$:
\begin{center}
\begin{tabular}{c|cccccc}
$d$ & 2 & 3 & 4 & 5 & 6 & 7 \\
\hline
$(d+1)\mu^2 - d$ & $1/4$ & $47/729$ & $0.013$ & $0.0017$ & $-0.0029$ & $-0.0055$
\end{tabular}
\end{center}
The sign changes between $d = 5$ (positive) and $d = 6$ (negative).

\emph{Large dimensions:} From the exact formula for $\mu$, expanding $(d+1)^{3/2} = d^{3/2}(1 + \frac{3}{2d} + O(d^{-2}))$ yields $\mu = 1 - 1/d + 2/d^{3/2} + O(d^{-2})$ and hence $\mu^2 = 1 - 2/d + 4/d^{3/2} + O(d^{-2})$, so
\begin{align}
(d+1)\mu^2 - d = d\mu^2 + \mu^2 - d = -1 + 4/\sqrt{d} + O(d^{-1}).
\end{align}
For $d \geq 17$, we have $4/\sqrt{d} < 1$, so $(d+1)\mu^2 - d < 0$. Combined with the numerical evaluation above, the condition fails for all integers $d \geq 6$.

Treating $d$ as a continuous variable, the crossover $d^*$ satisfies $(d^*+1)\mu(d^*)^2 = d^*$. Numerical root-finding gives $d^* \approx 5.70$.
\end{proof}

Table~\ref{tab:simplex-ratio} summarizes the quantum-to-classical ratio for small dimensions.

\begin{table}[ht]
\centering
\begin{tabular}{c|cccc}
$d$ & $\mu$ & $\omega_q$ & $\omega_c$ & $\omega_q/\omega_c$ \\
\hline
2 & $\sqrt{3}/2$ & $7/12$ & $5/9$ & $21/20 = 1.050$ \\
3 & $8/9$ & $209/324$ & $5/8$ & $1.032$ \\
4 & $0.9045$ & $0.6909$ & $0.68$ & $1.016$ \\
5 & $0.9160$ & $0.7260$ & $0.7222$ & $1.005$ \\
6 & $0.9248$ & $0.7537$ & $0.7551$ & $0.998$ \\
\end{tabular}
\caption{Quantum-to-classical ratio for the simplex strategy in $(d+1, 1)$-jamming games.}
\label{tab:simplex-ratio}
\end{table}

To understand this crossover, we analyze the asymptotic behavior of both quantities. For the classical value, Theorem~\ref{thm:classical-value} with $n = d+1$ and $k = 1$ gives
\begin{align}
    \omega_c(d+1, 1) = \frac{d^2 + 1}{(d+1)^2} = 1 - \frac{2d}{(d+1)^2} = 1 - \frac{2}{d} + O(d^{-2}). \label{eq:omega-c-asymp}
\end{align}

For the quantum value, we first expand the cross-context overlap $\mu$. Using $(d+1)^{3/2} = d^{3/2}(1 + \frac{3}{2d} + O(d^{-2}))$:
\begin{align}
    \mu &= \frac{d^3 - 3d - 2 + 2d^{3/2} + 3d^{1/2} + O(d^{-1/2})}{d^3 + d^2} = 1 - \frac{1}{d} + \frac{2}{d^{3/2}} + O(d^{-2}). \label{eq:mu-asymp}
\end{align}
Squaring and substituting into the expression for $\omega_q$:
\begin{align}
    \mu^2 &= 1 - \frac{2}{d} + \frac{4}{d^{3/2}} + O(d^{-2}), \\
    (d-1)\mu^2 &= d - 3 + \frac{4}{\sqrt{d}} + O(d^{-1/2}), \\
    \omega_q^{\mathrm{simplex}} &= \frac{d - 2 + \frac{4}{\sqrt{d}} + O(d^{-1/2})}{d+1} = 1 - \frac{3}{d} + \frac{4}{d^{3/2}} + O(d^{-2}). \label{eq:omega-q-asymp}
\end{align}

Comparing Eq.~\eqref{eq:omega-c-asymp} and Eq.~\eqref{eq:omega-q-asymp}:
\begin{align}
    \frac{\omega_q^{\mathrm{simplex}}}{\omega_c} = \frac{1 - \frac{3}{d} + O(d^{-3/2})}{1 - \frac{2}{d} + O(d^{-2})} = 1 - \frac{1}{d} + O(d^{-3/2}) < 1 \quad \text{for large } d. \label{eq:ratio-asymp}
\end{align}

The key insight is that while both $\omega_q$ and $\omega_c$ approach 1 as $d \to \infty$, the classical value approaches 1 \emph{faster}. The simplex strategy's quantum winning probability has a larger deficit from unity ($3/d$ versus $2/d$), leading to a ratio below 1.

Physically, this reflects the fact that the $(d+1, 1)$-jamming game becomes easy classically when $d$ is large: with only one channel jammed out of $d+1$, Alice and Bob's safe sets overlap in $d-1$ channels, and the greedy classical strategy coordinates efficiently. The simplex seed frame, constrained by its rigid geometric structure, cannot adapt to exploit this regime.

\begin{remark}[Implications for frame selection]
The crossover phenomenon highlights that optimal frame selection depends on the game parameters. The simplex, while providing the simplest analytical example, is not universally optimal. For $(d+1, 1)$-games with $d \geq 6$, one should either use a different seed frame or simply employ the classical strategy. This motivates the study of other frame constructions (harmonic, MUB, SIC-POVM) that may provide advantage in different parameter regimes.
\end{remark}

%% file: otherStrategies.tex
\subsection{Other Strategies}

We now consider three other possible choices of seed vectors: SIC-POVMS (Section~\ref{sec:sic}), Mutually Unbiased Bases (MUBs) (Section~\ref{sec:mub}), and AllTop Frames (Section~\ref{sec:alltop}). 
We numerically compare their performance, including to Harmonic and Simplex seed vectors, in Section~\ref{sec:numericalCompare}.

\subsubsection{A SIC-POVM Strategy}
\label{sec:sic}

We now consider seed vectors based on Symmetric Informationally Complete POVMs
(SIC-POVMs)~\cite{Zauner2011, Renes2004}. A SIC-POVM in dimension $d$ consists
of $d^2$ unit vectors $\{|\phi_c\rangle\}_{c=0}^{d^2-1}$ in $\mathbb{C}^d$
satisfying the equiangular property
\begin{align}
    |\langle\phi_a|\phi_b\rangle|^2 = \frac{1}{d+1} \quad \text{for all } a \neq b.
    \label{eq:sic-equiangular}
\end{align}
This is precisely the Welch bound~\cite{welch1974} for $n = d^2$ vectors in
$\mathbb{C}^d$, making SIC-POVMs equiangular tight frames~\cite{Waldron2018}.
Exact SIC-POVM constructions are known in all dimensions up to~53, as well as in
several higher dimensions~\cite{Scott2010, Appleby2018, Fuchs2017}; they are
conjectured to exist in every dimension $d \geq 2$~\cite{Zauner2011}.

Intuititvely, one would expect SIC-POVMs to perform very well in dimensions where they can be found: Since SIC-POVMs saturate the Welch bound, their coherence
$\delta_{\mathrm{SIC}} = 1/\sqrt{d+1}$ is the smallest achievable by any frame of
$d^2$ vectors in $\mathbb{C}^d$. Low coherence keeps the Gram matrices $G_X$ close
to the identity, so the L\"owdin orthonormalization acts as a mild perturbation of
the seed vectors. The orthonormalized vectors $|v^c_X\rangle$ therefore retain
significant overlap with the original $|\phi_c\rangle$, preserving the cross-context
alignment that drives large $\operatorname{tr}[A_c^2]$ in
Eq.~\eqref{eq:inTermsOfA}. Moreover, the equiangular property ensures this
favorable conditioning is uniform across all safe sets---unlike, say, harmonic frames
(Section~\ref{sec:harmonic}), where the overlap magnitude depends on the channel
separation.

Mapping to our jamming game, SIC-POVMs yield $n = d^2$ channels with safe set size
$d$, corresponding to $k = d(d-1)$ jammed channels. More generally, a SIC-POVM can
serve as the seed frame for any $(n,k)$-jamming game with $n \leq d^2$ and $k = n - d$.
In our numerical studies below, we restrict ourselves to the case $n = d^2$.

\paragraph{Numerical quantum advantage.}
We numerically evaluate the quantum winning probability for SIC-POVM seed frames with $n = d^2$ channels. Table~\ref{tab:sic-advantage} reports the results for $d = 2, 3, 4$.

\begin{table}[ht]
\centering
\begin{tabular}{ccc|ccc}
\toprule
$d$ & $n$ & $k$ & $\omega_q^{\mathrm{SIC}}$ & $\omega_c$ & $\omega_q/\omega_c$ \\
\midrule
2 & 4 & 2 & 0.4167 & 0.3889 & 1.071 \\
3 & 9 & 6 & 0.2343 & 0.2262 & 1.036 \\
4 & 16 & 12 & 0.1675 & 0.1580 & 1.060 \\
\bottomrule
\end{tabular}
\caption{Quantum winning probability for SIC-POVM seed frames with $n = d^2$ channels.}
\label{tab:sic-advantage}
\end{table}

In all cases tested, the SIC-POVM strategy achieves quantum advantage, with improvement ranging from $3.6\%$ ($d = 3$) to $7.1\%$ ($d = 2$). At $(n, d) = (4, 2)$, the SIC-POVM value $\omega_q = 5/12$ matches the NPA hierarchy bound (Section~\ref{sec:numericalCompare}), certifying this as the true quantum value of the $(4, 2)$-jamming game.

\subsubsection{A MUB Strategy}
\label{sec:mub}

We now consider seed vectors constructed from mutually unbiased bases
(MUBs)~\cite{Schwinger1960, WoottersFields1989}. Two orthonormal bases
$\mathcal{B}_1 = \{|e_i\rangle\}_{i=1}^{d}$ and
$\mathcal{B}_2 = \{|f_j\rangle\}_{j=1}^{d}$ in $\mathbb{C}^d$ are called
\emph{mutually unbiased} if
\begin{align}
    |\langle e_i | f_j \rangle|^2 = \frac{1}{d} \quad \text{for all } i, j.
    \label{eq:mub_condition}
\end{align}
This concept was introduced by Schwinger~\cite{Schwinger1960} and later formalized
by Wootters and Fields~\cite{WoottersFields1989} (see Durt
et al.~\cite{DurtEnglertBengtssonZyczkowski2010} for a comprehensive review). The
maximum number of pairwise mutually unbiased bases in dimension~$d$ is at most
$d + 1$, and this bound is achieved whenever $d = p^r$ is a prime
power~\cite{WoottersFields1989, BandyopadhyayBoykinRoychowdhuryVatan2002}. We restrict ourselves to this case, and consider a full set of MUBs.
We remark an interesting connection to classical communications: discrete chirp
sequences, used since the 1960s for jamming-resistant spread-spectrum
communication~\cite{WinklerChirp1962, SpringerSpread2000}, are instances of MUB
constructions in prime dimensions~\cite{BandyopadhyayBoykinRoychowdhuryVatan2002,
KlappeneckerRoetteler2004}.

When a complete set of $d+1$ MUBs exists, we can construct our seed frame by taking the
union of all bases:
\begin{align}
    \mathcal{F}_{\mathrm{MUB}} = \bigcup_{m=0}^{d} \mathcal{B}_m, \quad
    \text{yielding } n = d(d + 1) \text{ seed vectors}.
    \label{eq:mub_seed_frame}
\end{align}
The coherence structure of this frame has two levels: vectors within the same basis are
orthogonal, while vectors from different bases satisfy
Eq.~\eqref{eq:mub_condition}. The coherence is therefore
$\delta_{\mathrm{MUB}} = 1/\sqrt{d}$, which is close to, but does not exactly achieve, the Welch
bound~\cite{welch1974} $\delta_{\mathrm{Welch}} = \sqrt{d/(d^2+d-1)}$ for $d(d+1)$ vectors in $\mathbb{C}^d$
(the ratio $\delta_{\mathrm{MUB}}/\delta_{\mathrm{Welch}} = \sqrt{1 + (d-1)/d^2} \to 1$ as $d \to \infty$).
Nonetheless, the low coherence keeps the Gram matrices $G_X$
well-conditioned, so the L\"owdin orthonormalization preserves cross-context alignment.

Mapping to our jamming game, the full MUB frame yields $n = d(d+1)$ channels with
safe set size~$d$, corresponding to $k = d^2$ jammed channels. More generally, a
subset of the MUB vectors can serve as the seed frame for any $(n,k)$-jamming game
with $n \leq d(d+1)$ and $k = n - d$. In our numerical studies below, we restrict
ourselves to the case $n = d(d+1)$, i.e., we use the full set of MUBs in
dimension~$d$.

\paragraph{Numerical quantum advantage.}
We numerically evaluate the quantum winning probability for MUB seed frames with $n = d(d+1)$ channels. Table~\ref{tab:mub-advantage} reports the results for $d = 2, 3, 4$.

\begin{table}[ht]
\centering
\begin{tabular}{ccc|ccc}
\toprule
$d$ & $n$ & $k$ & $\omega_q^{\mathrm{MUB}}$ & $\omega_c$ & $\omega_q/\omega_c$ \\
\midrule
2 & 6 & 4 & 0.2644 & 0.2444 & 1.082 \\
3 & 12 & 9 & 0.1698 & 0.1641 & 1.035 \\
4 & 20 & 16 & 0.1296 & 0.1237 & 1.048 \\
\bottomrule
\end{tabular}
\caption{Quantum winning probability for MUB seed frames with $n = d(d+1)$ channels.}
\label{tab:mub-advantage}
\end{table}

The MUB strategy achieves the largest quantum advantage among all explicit frame constructions tested: an $8.2\%$ improvement over the classical value at $(n, d) = (6, 2)$. In all cases, the MUB strategy provides quantum advantage, and numerical optimization (Section~\ref{sec:numericalCompare}) confirms that the MUB frame is optimal at $(n, d) = (6, 2)$.

\subsubsection{An AllTop Strategy}
\label{sec:alltop}

Finally, we consider seed vectors constructed from AllTop sequences~\cite{Alltop1980},
which were originally introduced in the context of radar and spread-spectrum
communications. The AllTop construction exploits the arithmetic structure of finite
fields to produce vectors with controlled inner products. Unlike the MUB construction
of Section~\ref{sec:mub}, which organizes vectors into mutually unbiased \emph{bases}
with orthogonality within each basis, AllTop vectors form a frame with no orthogonal
pairs but with uniformly bounded pairwise overlaps. AllTop frames are known to exist
whenever $n \geq 5$ is prime, with generalizations to prime powers~\cite{Ding2007}.

Let $n \geq 5$ be prime and let $d \leq n$ be a positive integer. Let
$\omega = e^{2\pi i/n}$ denote a primitive $n$-th root of unity. For each channel
$c \in \{0, 1, \ldots, n-1\}$, define the seed vector
$|\phi_c\rangle \in \mathbb{C}^d$ by
\begin{align}\label{eq:alltop-def}
    |\phi_c\rangle = \frac{1}{\sqrt{d}} \sum_{j=0}^{d-1} \omega^{c j^3} \ket{j}.
\end{align}
This construction can be used in an $(n, n-d)$-jamming game for any prime $n \geq 5$ and any
safe set size $d \leq n$. The key advantage is universality: unlike simplex frames
(restricted to $k = 1$), MUB frames (restricted to $n = d(d+1)$ with $d$ a prime
power), or SIC-POVM frames (restricted to $n = d^2$), AllTop frames exist for all
prime $n \geq 5$ with arbitrary $d \leq n$, imposing no arithmetic constraints between
the game parameters $n$ and~$d$.

\paragraph{Numerical quantum advantage.}
Unlike the equiangular frames above, the AllTop frame exhibits highly variable performance. For $d = 2$, it achieves a consistent $\sim\!6\%$ advantage across all prime $n$ tested (e.g., $\omega_q/\omega_c \approx 1.061$ at $(n, d) = (5, 2)$). However, for $d \geq 3$ the performance is erratic: certain $(n, d)$ combinations yield substantial quantum advantage (e.g., $+6.9\%$ at $(n, d) = (19, 3)$), while others lead to winning probabilities far below the classical value (e.g., $\omega_q/\omega_c \approx 0.28$ at $(n, d) = (7, 6)$). Of 28 $(n, d)$ cases tested, 17 show quantum advantage, but 11 show the classical strategy outperforming the AllTop quantum strategy. This instability arises because, unlike SIC-POVMs and MUBs, the AllTop inner product magnitudes depend on the channel separation $c_2 - c_1 \pmod{n}$ through the cubic phase structure, and can cause near-linear dependence in certain safe sets, severely distorting the L\"owdin orthonormalization.

\subsubsection{Numerical Comparison}
\label{sec:numericalCompare}

We now compare the quantum winning probabilities achieved by all explicit frame constructions, and relate them to upper bounds from the SDP hierarchies~\cite{NPA2008, DLTW2008} and the non-signaling (NS) value. Software for the numerical
investigations can be found at~\cite{JammingNumerics}.

\paragraph{Comparison of explicit frames.}
Table~\ref{tab:frame-comparison-SM} compares the quantum winning probabilities for all frame constructions and the numerically optimized seed frame across small $(n, d)$ cases. The optimized column reports the result of numerical optimization over all seed frames in $\mathbb{C}^d$ using L-BFGS-B with multiple random restarts (see Section~\ref{sec:seedAnsatz} for details on the methodology).

\begin{table}[ht]
\centering
\small
\begin{tabular}{ccc|c|ccccc|c}
\toprule
$n$ & $d$ & $k$ & $\omega_c$ & Simplex & Harmonic & MUB & SIC & AllTop & Opt. \\
\midrule
3 & 2 & 1 & 0.5556 & \textbf{0.5833} & 0.5833 & -- & -- & -- & 0.5833 \\
4 & 2 & 2 & 0.3889 & -- & 0.4119 & -- & \textbf{0.4167} & -- & 0.4167 \\
5 & 2 & 3 & 0.3000 & -- & 0.3184 & -- & -- & 0.3184 & \textbf{0.3230} \\
6 & 2 & 4 & 0.2444 & -- & 0.2595 & \textbf{0.2644} & -- & -- & 0.2644 \\
\midrule
4 & 3 & 1 & 0.6250 & \textbf{0.6451} & 0.6451 & -- & -- & -- & 0.6451 \\
5 & 3 & 2 & 0.4600 & -- & 0.4767 & -- & -- & 0.4767 & \textbf{0.4808} \\
6 & 3 & 3 & 0.3650 & -- & 0.3783 & -- & -- & -- & \textbf{0.3852} \\
\midrule
5 & 4 & 1 & 0.6800 & \textbf{0.6909} & 0.6909 & -- & -- & 0.6909 & 0.6909 \\
6 & 4 & 2 & 0.5200 & -- & 0.5290 & -- & -- & -- & \textbf{0.5333} \\
\midrule
6 & 5 & 1 & 0.7222 & \textbf{0.7260} & 0.7260 & -- & -- & -- & 0.7260 \\
7 & 6 & 1 & 0.7551 & 0.7537 & 0.7537 & -- & -- & 0.2097 & \textbf{0.7537} \\
\bottomrule
\end{tabular}
\caption{Comparison of quantum winning probabilities for different seed frame constructions. Bold indicates the best value for each case. A dash (--) indicates the frame does not exist for those parameters. The last column (Opt.)\ reports the result of numerical optimization over all seed frames. At $(n, d) = (7, 6)$, even the optimized quantum value falls below $\omega_c$.}
\label{tab:frame-comparison-SM}
\end{table}

Figure~\ref{fig:strategy-grid} visualizes the best-performing explicit frame construction at each $(n, d)$ point.

\begin{figure}[ht]
\centering
\includegraphics[width=0.85\textwidth]{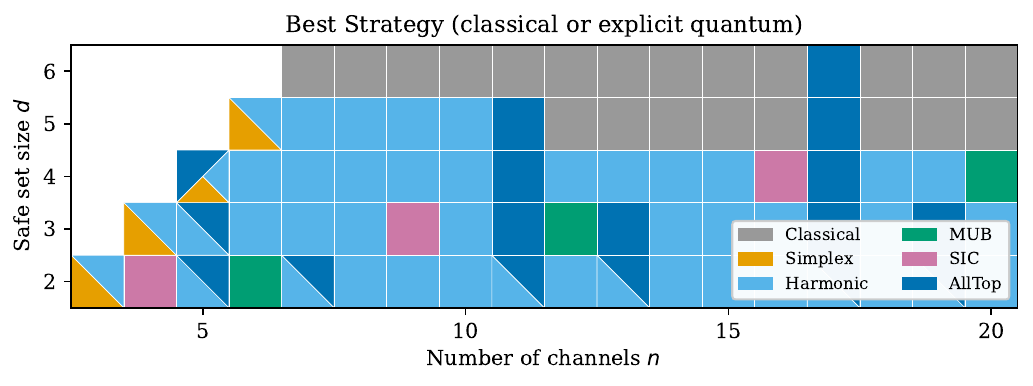}
\caption{Best explicit quantum strategy at each $(n, d)$ point. Each cell is colored by the frame achieving the highest quantum winning probability: Simplex (available only at $n = d+1$, i.e.\ $k = 1$), Harmonic (available for all $n \geq d$), MUB (available at $n = d(d+1)$ when $d$ is a prime power), SIC-POVM (available at $n = d^2$ for known dimensions), and AllTop (available for prime $n \geq 5$ with any $d \leq n$). The harmonic frame dominates the interior of the grid due to its universality, while the structured frames (MUB, SIC, simplex) are optimal at their special parameter points.}
\label{fig:strategy-grid}
\end{figure}

Several patterns are apparent from Table~\ref{tab:frame-comparison-SM}:
\begin{itemize}
\item The simplex, MUB, and SIC frames are optimal wherever they exist: numerical optimization finds no improvement over these frames at their respective parameter values $(n, d) = (d+1, d)$, $(d(d+1), d)$, and $(d^2, d)$.
\item The harmonic frame, while available for all $(n, d)$, is suboptimal in most cases---typically by $0.3\text{--}1.8\%$ compared to the optimized seed frame. Its main value lies in universality: it provides quantum advantage for all $d \leq 4$ without requiring special values of $n$.
\item The AllTop frame performs well for $d = 2$ but is unreliable for $d \geq 3$. At $(n, d) = (7, 6)$, the AllTop strategy is dramatically worse than classical ($\omega_q/\omega_c \approx 0.28$).
\item At $d = 6$ with $k = 1$, neither the simplex nor harmonic frame achieves quantum advantage, confirming the crossover phenomenon of Proposition~\ref{prop:simplex-crossover}. Even numerical optimization over all seed frames yields $\omega_q^{\mathrm{opt}} = 0.7537 < 0.7551 = \omega_c$.
\end{itemize}

\paragraph{Upper bounds: SDP hierarchy and non-signaling value.}
To assess how close our strategies come to the true quantum value, we computed upper bounds from the first level of the Navascu\'es--Pironio--Ac\'in (NPA) hierarchy~\cite{NPA2008} (or its dual put forward by Doherty, Liang, Toner and Wehner~\cite{DLTW2008}) and the non-signaling (NS) value via linear programming. Table~\ref{tab:npa-comparison} reports these bounds for cases where computation was feasible.

\begin{table}[ht]
\centering
\begin{tabular}{ccc|cccc}
\toprule
$n$ & $d$ & $k$ & $\omega_c$ & $\omega_q^{\mathrm{opt}}$ & NPA$_1$ & $\omega_{\mathrm{NS}}$ \\
\midrule
3 & 2 & 1 & 0.5556 & 0.5833 & 0.5833 & 0.6667 \\
4 & 2 & 2 & 0.3889 & 0.4167 & 0.4167 & 0.5000 \\
5 & 2 & 3 & 0.3000 & 0.3230 & 0.3250 & 0.4000 \\
4 & 3 & 1 & 0.6250 & 0.6451 & 0.6944 & 0.7500 \\
5 & 3 & 2 & 0.4600 & 0.4808 & 0.5333 & 0.6000 \\
5 & 4 & 1 & 0.6800 & 0.6909 & 0.7625 & 0.8000 \\
\bottomrule
\end{tabular}
\caption{Comparison of the optimized quantum value with the NPA level-1 upper bound and the non-signaling value $\omega_{\mathrm{NS}}$.}
\label{tab:npa-comparison}
\end{table}

For the two smallest cases---$(n, d) = (3, 2)$ and $(4, 2)$---the optimized quantum value matches the NPA$_1$ bound to machine precision, certifying these as the exact quantum values of the respective jamming games. The non-signaling value significantly exceeds the quantum value in all cases (by $14\text{--}16\%$), indicating that the jamming game exhibits a gap between quantum and non-signaling correlations analogous to what is observed in other nonlocal games.

For $d \geq 3$, the NPA$_1$ bound becomes loose: the gap between $\omega_q^{\mathrm{opt}}$ and NPA$_1$ ranges from $5\text{--}10\%$, suggesting that higher levels of the NPA hierarchy or tailored relaxations would be needed to obtain tight bounds. This is consistent with the general observation that NPA level~1 provides tight bounds primarily for games with small output alphabets.

\subsubsection{Is the seed vector Ansatz restrictive?}
\label{sec:seedAnsatz}

Our quantum strategies are constructed by choosing $n$ seed vectors in $\mathbb{C}^d$ and applying L\"owdin orthonormalization within each safe set (Section~\ref{sec:lowdin}). A natural question is whether this Ansatz is restrictive: could one achieve higher winning probabilities by optimizing directly over all rank-1 measurement operators (vectors), without the constraint that they arise from a common seed frame?

To investigate this, we performed two parallel numerical optimizations for each $(n, d)$ pair with $|\mathcal{S}| \leq 30$:
\begin{itemize}
\item Seed optimization: Optimize over $n$ unit vectors $\{|\phi_c\rangle\}_{c \in [n]}$ in $\mathbb{C}^d$, with measurements obtained by L\"owdin orthonormalization. The search space has $n(2d - 1)$ real parameters.

\item General rank-1 optimization: Optimize directly over all measurement vectors $\{|v_x^c\rangle\}$, subject to orthonormality within each safe set and the synchronicity constraint (Bob uses conjugated bases). This search space is significantly larger.
\end{itemize}

Both optimizations use L-BFGS-B with 10 random restarts, convergence tolerance $10^{-10}$, and 1000 maximum iterations per restart. Table~\ref{tab:ansatz-comparison} reports the results.

\begin{table}[ht]
\centering
\begin{tabular}{cc|cc|c}
\toprule
$n$ & $d$ & $\omega_q^{\mathrm{seed}}$ & $\omega_q^{\mathrm{general}}$ & Gap \\
\midrule
3 & 2 & 0.583333 & 0.583333 & $<\!10^{-9}$ \\
4 & 2 & 0.416667 & 0.416667 & $<\!10^{-9}$ \\
5 & 2 & 0.322956 & 0.322958 & $2 \times 10^{-6}$ \\
6 & 2 & 0.264379 & 0.264379 & $<\!10^{-9}$ \\
4 & 3 & 0.645062 & 0.645062 & $<\!10^{-9}$ \\
5 & 3 & 0.480808 & 0.480854 & $5 \times 10^{-5}$ \\
6 & 3 & 0.385211 & 0.385216 & $5 \times 10^{-6}$ \\
5 & 4 & 0.690881 & 0.690881 & $<\!10^{-9}$ \\
6 & 4 & 0.533262 & 0.533275 & $1 \times 10^{-5}$ \\
6 & 5 & 0.725987 & 0.725987 & $<\!10^{-9}$ \\
7 & 6 & 0.753706 & 0.753706 & $<\!10^{-9}$ \\
\bottomrule
\end{tabular}
\caption{Comparison of seed-based (L\"owdin) and general rank-1 optimization. The gap is $\omega_q^{\mathrm{general}} - \omega_q^{\mathrm{seed}}$. In 7 of 11 cases the gap is below $10^{-9}$; in the remaining cases the discrepancy is at most $5 \times 10^{-5}$, within the tolerance expected from nonconvex optimization with finite restarts.}
\label{tab:ansatz-comparison}
\end{table}

For the parameter ranges studied, the seed vector Ansatz appears to incur no loss relative to unconstrained rank-1 optimization. In 7 of 11 cases, the gap between the two optimizations is below $10^{-9}$, indicating identical optima to numerical precision. In the remaining cases, the discrepancies are at most $5 \times 10^{-5}$, well within the tolerance expected from nonconvex optimization with finite restarts.

This observation is noteworthy but should be interpreted with caution. A priori, the seed Ansatz constrains measurements across different safe sets to be related through a common set of seed vectors, while the general rank-1 approach allows each safe set's measurement basis to be chosen independently. For the \emph{small} $(n, d)$ values accessible to numerical optimization, the Ansatz matches unconstrained optimization, and for the two cases where the SDP bound is tight ($(n, d) = (3, 2)$ and $(4, 2)$) the seed optimization achieves the certified quantum optimum. Whether this pattern persists for larger parameters remains an open question.